\documentclass[12pt, notitlepage]{article}

\usepackage{a4}

\usepackage{enumitem}

\usepackage{amssymb, graphicx}
\usepackage{subfigure}
\usepackage{amsmath}
\usepackage{undertilde}
\usepackage{verbatim}
\usepackage{bbm}
\usepackage{ dsfont }
\usepackage{geometry}
\usepackage{amsthm}
\newtheorem{thm}{Theorem}[section]
\newtheorem{lem}[thm]{Lemma}
\newtheorem{cor}[thm]{Corollary}
\newtheorem{exa}[thm]{Example}
\newtheorem{rem}[thm]{Remark}

\usepackage{url}

\usepackage{natbib}

\usepackage{sectsty}
\allsectionsfont{\sffamily}
\usepackage{caption}
\newcommand{\N}{\mathbb{N}}
\newcommand{\Prob}{\mathbb{P}}
\newcommand{\R}{\mathbb{R}}

\newcommand{\Mb}{\mathbf{M}}

\newcommand{\Ub}{\mathbf{U}}
\newcommand{\Vb}{\mathbf{V}}

\newcommand{\Wb}{\mathbf{W}}

\newcommand{\Zb}{\mathbf{Z}}

\newcommand{\be}{\begin{equation}}
\newcommand{\ee}{\end{equation}}

\begin{document}
{
\renewcommand*{\thefootnote}{\fnsymbol{footnote}}
\title{\textbf{\sffamily Block-Maxima of Vines
}}
\date{\small January 2015}
\author{Matthias Killiches\footnote{Zentrum Mathematik, Technische Universit\"at M\"unchen, Boltzmannstra\ss e 3, 85748 Garching, Germany (email: \texttt{matthias.killiches@tum.de}). Corresponding author.} { and Claudia Czado\footnote{Zentrum Mathematik, Technische Universit\"at M\"unchen, Boltzmannstra\ss e 3, 85748 Garching, Germany (email: \texttt{cczado@ma.tum.de}).}}}

\maketitle

\begin{abstract}
We examine the dependence structure of finite block-maxima of multivariate distributions. We provide a closed form expression for the copula density of the vector of the block-maxima. Further, we show how partial derivatives of three-dimensional vine copulas can be obtained by only one-dimensional integration. Combining these results allows the numerical treatment of the block-maxima of any three-dimensional vine copula for finite block-sizes. We look at certain vine copula specifications and examine how the density of the block-maxima behaves for different block-sizes. Additionally, a real data example from hydrology is considered. In extreme-value theory for multivariate normal distributions, a certain scaling of each variable and the correlation matrix is necessary to obtain a non-trivial limiting distribution when the block-size goes to infinity. This scaling is applied to different three-dimensional vine copula specifications.\\

\noindent \textit{Keywords: Multivariate copula, vine copulas, finite block-maxima, scaled block-maxima, extreme-value scaling.}
\end{abstract}
}
\section{Copula Density of the Distribution of Block-Maxima}
Basically, block-maxima have been used in extreme-value theory as one approach to derive the family of General Extreme-Value (GEV) distributions (\cite{mcneil2010quantitative}). In the recent past the block-maxima method has been studied more thoroughly and compared to the peaks-over-threshold (POT) method in \cite{ferreira2014block} and \cite{jaruvskova2006peaks}. \cite{dombry2013maximum} justifies the usage of a maximum-likelihood estimator for the extreme-value index within the block-maxima framework. The numerical convergence of the block-maxima approach to the GEV distribution is examined in \cite{faranda2011numerical}. Moreover, the block-maxima method has found its way into many application areas: \cite{marty2012long} investigate long-term changes in annual maximum snow depth and snowfall in Switzerland. Temperature, precipitation, wind extremes over Europe are analyzed in \cite{nikulin2011evaluation}. A spatial application can be found in \cite{naveau2009modelling}. \cite{rocco2014extreme} provides an overview over the concepts of extreme-value theory being used in finance. While many of the articles use univariate concepts, \cite{bucher2013extreme} treats how to estimate extreme-value copulas based on block-maxima of a multivariate stationary time series. In contrast to the existing literature (known to the authors), in the following we will consider finite block-maxima of multivariate random variables focusing on the dependence structure.\\

Let $\Ub=(U_1,\ldots,U_d)'$ be a random vector with $\mathcal{U}[0,1]$-distributed margins, copula $C$ and copula density $c$. We consider $n$ i.i.d. copies $\Ub_i=(U_{i,1},\ldots,U_{i,d})'$ of $\Ub$, $i=1, \ldots,n$. We apply the inverse probability integral transform to each component of $\Ub_i$ to obtain marginally normalized data (called z-scale):
\[
\Zb_i=(Z_{i,1},\ldots,Z_{i,d})' \text{ with } Z_{i,j}:=\Phi^{-1}(U_{i,j})\sim \mathcal{N}(0,1),
\]
for $i=1,\ldots,n$, $j=1,\ldots,d$, where $\Phi^{-1}$ is the quantile function of the standard normal distribution $\mathcal{N}(0,1)$. We consider this normalized scale since later we want to compare this to the limiting approach used to derive the multivariate H\"{u}sler-Reiss extreme-value copula. We are interested in the distribution $F^{(n)}$ of the vector of componentwise block-maxima 
\[
\Mb^{(n)}=(M^{(n)}_{1},\ldots,M^{(n)}_{d})' \text{ with } M^{(n)}_{j}:=\max_{i=1,\ldots,n}Z_{i,j}
\] 
for $j=1,\ldots,d$. According to \cite{Sklar} the dependency structure is determined by the corresponding copula $C_{\Mb^{(n)}}$. Since $Z_{i,j}$, $i=1,\ldots,n$, are i.i.d., we know that the marginal distribution functions of $M^{(n)}_{j}$ are given by
\be\label{eq:Fnj}
F^{(n)}_{j}(m_j)=\Prob\left(Z_{1,j}\leq m_j,\ldots,Z_{n,j}\leq m_j\right)=\Phi(m_j)^n
\ee
and hence the corresponding densities are
\be\label{eq:fnj}
f^{(n)}_{j}(m_j)=n\Phi(m_j)^{n-1}\varphi(m_j)
\ee
for $m_j\in \R$, $j=1,\ldots,d$. Here $\Phi$ and $\varphi$ denote the distribution function and the density of the standard normal distribution, respectively. Thus, the copula $C_{\Mb^{(n)}}$ is the distribution function of
\[
\Vb=(V_1,\ldots,V_d)' \text{ with } V_j:=\Phi\left(M_j^{(n)}\right)^n\sim \mathcal{U}[0,1].
\]
For $n\in\N$ the copula of the componentwise maxima $C_{\Mb^{(n)}}$ can be expressed in terms of the underlying copula $C$ as follows
\be\label{eq:maxcop}
C_{\Mb^{(n)}}(u_1,\ldots,u_d)=C\left(u_1^{1/n},\ldots,u_d^{1/n}\right)^n,
\ee
where $u_j\in [0,1]$, $j=1,\ldots,d$. Since $C$ is assumed to have a density $c$, Equation \ref{eq:maxcop} yields that $C_{\Mb^{(n)}}$ also has a density, denoted by $c_{\Mb^{(n)}}$. Using Sklar's Theorem, Equations \ref{eq:Fnj} and \ref{eq:maxcop} imply that the joint distribution function of $\Mb^{(n)}$ is given by
\[
F^{(n)}(m_1,\ldots,m_d)=C_{\Mb^{(n)}}\left(F^{(n)}_{1}(m_1),\ldots,F^{(n)}_{d}(m_d)\right)=C\left(\Phi(m_1),\ldots,\Phi(m_d)\right)^n.
\]
\begin{thm}\label{thm:copuladensity}
The density of the copula of the vector of block-maxima satisfies for $u_j\in[0,1]$, $j=1\ldots,d$:
\begin{multline}\label{eq:copuladensity}
c_{\Mb^{(n)}}(u_1,\ldots,u_d) =
\frac{1}{n^d}\left(\prod_{j=1}^d u_j\right)^{\frac{1}{n}-1}\sum_{j=1}^{d \wedge n}\Bigg\{\frac{n!}{(n-j)!} C\left(u_1^{1/n},\ldots,u_d^{1/n}\right)^{n-j}\Bigg.\\
\Bigg. \cdot\sum_{\mathcal{P}\in\mathcal{S}_{d,j}} \prod_{M\in\mathcal{P}} \partial_M C\left(u_1^{1/n},\ldots,u_d^{1/n}\right) \Bigg\},
\end{multline}
where
$d \wedge n:=\min\left\{d,n\right\}$, $\mathcal{S}_{d,j}:=\left\{\mathcal{P} | \mathcal{P}\textrm{  partition of }\left\{1,\ldots,d\right\} \textrm{ with } \left|\mathcal{P}\right|=j \right\}$ and 
\[
\partial_M C\left(u_1^{1/n},\ldots,u_d^{1/n}\right):=\frac{\partial^p C(v_1,\ldots,v_d)}{\partial v_{m_1} \cdots \partial v_{m_p}}\bigg|_{v_1=u_1^{1/n},\ldots, v_d=u_d^{1/n}}
\]
for $M=\left\{m_1,\ldots,m_p\right\}\subseteq \left\{1,\ldots,d\right\}$.
\end{thm}
The proof of Theorem \ref{thm:copuladensity} as well as all other proofs can be found in the appendix at the end of this chapter (pages \pageref{sec:appendix} ff.).

For the joint density $f^{(n)}$ of the block-maxima with marginally normalized data (on the z-scale) we also obtain an explicit expression.
\begin{cor}
\label{cor:maxzden}
For $m_j\in\R$, $j=1,\ldots,d$, we have
\begin{multline}
f^{(n)}(m_1,\ldots,m_d)=\left(\prod_{j=1}^d \varphi(m_j) \right)\cdot\sum_{j=1}^{d \wedge n}\Bigg\{ \frac{n!}{(n-j)!}\cdot C\left(\Phi(m_1),\ldots,\Phi(m_d)\right)^{n-j}\big. \\
\big. \cdot \sum_{\mathcal{P}\in\mathcal{S}_{d,j}} \prod_{M\in\mathcal{P}} \partial_M C\left(\Phi(m_1),\ldots,\Phi(m_d)\right) \Bigg\}.
\end{multline}
\end{cor}

\begin{exa}
Let $d=3$, $n\in\N$, i.e. $\Ub_i=(U_{i,1},U_{i,2},U_{i,3})'$, $\Zb_i=(Z_{i,1},Z_{i,2},Z_{i,3})'$ and $\Mb^{(n)}=(M^{(n)}_1,M^{(n)}_2,M^{(n)}_3)'$, $i=1,\ldots,n$. If $n\geq 3$ the copula density of the vector of the block-maxima is given by
\[
\begin{split}
c_{\Mb^{(n)}}(u_1,&u_2,u_3)=\frac{(u_1u_2u_3)^{\frac{1}{n}-1}}{n^3}\cdot \Big\{nC\left(u_1^{1/n},u_2^{1/n},u_3^{1/n}\right)^{n-1}c\left(u_1^{1/n},u_2^{1/n},u_3^{1/n}\right)\Big.\\
&\quad+n(n-1)C\left(u_1^{1/n},u_2^{1/n},u_3^{1/n}\right)^{n-2}\\
&\quad\cdot\Big[\partial_1C\left(u_1^{1/n},u_2^{1/n},u_3^{1/n}\right)\partial_{23}C\left(u_1^{1/n},u_2^{1/n},u_3^{1/n}\right)\Big.\\
&\quad+\partial_2C\left(u_1^{1/n},u_2^{1/n},u_3^{1/n}\right)\partial_{13}C\left(u_1^{1/n},u_2^{1/n},u_3^{1/n}\right)\\
&\Big.\quad+\partial_3C\left(u_1^{1/n},u_2^{1/n},u_3^{1/n}\right)\partial_{12}C\left(u_1^{1/n},u_2^{1/n},u_3^{1/n}\right)\Big]\\
&\quad+n(n-1)(n-2)C\left(u_1^{1/n},u_2^{1/n},u_3^{1/n}\right)^{n-3}\partial_1C\left(u_1^{1/n},u_2^{1/n},u_3^{1/n}\right)\\
&\quad\Big.\cdot\partial_2C\left(u_1^{1/n},u_2^{1/n},u_3^{1/n}\right)\partial_3C\left(u_1^{1/n},u_2^{1/n},u_3^{1/n}\right)
\Big\}.
\end{split}
\]
\end{exa}

\section{Vine Copulas}
While the catalog of bivariate copula families (see for example \cite{Joe}) is large this is not the case for multivariate copula families. They were initially 
dominated by Archimedean and elliptical copulas, however for complex dependency patterns such as asymmetric dependence in the tails these classes are insufficient. The class of vine copulas (\cite{Bedford}, \cite{vinebook}, \cite{aasczado}, \cite{kuro:joe:2010}) can accommodate such patterns. See \cite{stoeberczado2012} for a tutorial introduction and \cite{cza:2010} and \cite{jaworski2012} for recent reviews. Basically vine copulas are constructed using bivariate copulas called pair copulas as building blocks which are combined to form multivariate copulas using conditioning. The pair copulas represent the copula associated with bivariate conditional distributions. The conditioning variables are determined with the help of a sequence of linked trees called the vine structure. Further it is commonly assumed that the conditioning value does not influence the copula and its parameter.  See \cite{stoeber-vines} for a discussion of this simplifying condition. Further they show that multivariate Clayton copula is the only Archimedean copula which can be represented as a vine copula, while the multivariate t-copula is the only scale elliptical one. For the multivariate Gaussian copula and the t-copula the needed pair copulas are bivariate Gaussian or t-copulas, respectively. The corresponding parameters are given by (partial) correlation parameters. 

Vine copulas allow for product expressions of the
density. We only consider three-dimensional vine copulas which can be expressed as
\begin{equation}\label{eq:pcc}
	c(u_1,u_2,u_3) = c_{12}(u_1,u_2) c_{23}(u_2,u_3) c_{13;2}(C_{1|2}(u_1|u_2),C_{3|2}(u_3|u_2)).
\end{equation}
Here $c_{ij;k}$ denotes the bivariate copula density corresponding to bivariate distribution $(U_i,U_j)$ given $U_k=u_k$ and $C_{i|j}(u_i|u_j)$
denotes the conditional distribution function of $U_i$ given $U_j=u_j$, which 
can be expressed as
\[
C_{i|j}(u_i|u_j)= \frac{ \partial }{\partial u_j}  C_{ij}(u_i,u_j)
\] 
Further we write the bivariate copula densities in terms of their copula, i.e.
$
	c_{ij}(u_i,u_j) = \frac{\partial^2}{\partial u_i \partial u_j} C_{ij}(u_i,u_j).
$
For the  pair copulas $C_{12}, C_{23}, C_{13;2}$ arbitrary bivariate copulas can be utilized. Many bivariate families including rotations are implemented
in the R library \texttt{VineCopula} (see \cite{VC}), which allows for parameter estimation and model selection of vine copulas in arbitrary dimensions.

Now we will consider the three-dimensional case and derive expressions for the partial derivatives needed in Theorem \ref{thm:copuladensity} for the expression of the copula density for the block-maxima.

\begin{thm}\label{thm:copuladeriv}
For the vine copula density \eqref{eq:pcc} we have:
\begin{enumerate}
	\item $C(u_1,u_2,u_3) = \int_0^{u_2} C_{13;2}\left( C_{1|2}(u_1|v_2), C_{3|2}(u_3|v_2)\right) dv_2 $,
	\item
\begin{enumerate}
	\item $\partial_{1}C(u_1,u_2,u_3) = \int_0^{u_2} \partial_{1} C_{13|2}(C_{1|2}(u_1|v_2),C_{3|2}(u_3|v_2)) c_{12}(u_1,v_2) dv_2$,
	\item $\partial_{2} C(u_1,u_2,u_3) = C_{13;2}(C_{1|2}(u_1|u_2),C_{3|2}(u_3|u_2))$,
	\item $\partial_{3} C(u_1,u_2,u_3) = \int_0^{u_2} \partial_{3} C_{13|2}(C_{1|2}(u_1|v_2),C_{3|2}(u_3|v_2)) c_{23}(v_2,u_3) dv_2$,
\end{enumerate}
\item
\begin{enumerate}
	\item $\partial_{12} C(u_1,u_2,u_3) = \partial_{1} C_{13;2}(C_{1|2}(u_1|u_2),C_{3|2}(u_3|u_2)) c_{12}(u_1,u_2)$,
	\item $\partial_{13} C(u_1,u_2,u_3) = \int_0^{u_2} c_{13;2}(C_{1|2}(u_1|v_2),C_{3|2}(u_3|v_2)) c_{23}(v_2,u_3)c_{12}(u_1,v_2) dv_2$,
	\item $\partial_{23} C(u_1,u_2,u_3) = \partial_{3} C_{13;2}(C_{1|2}(u_1|u_2),C_{3|2}(u_3|u_2)) c_{23}(u_2,u_3)$,
\end{enumerate}
\item $c(u_1,u_2,u_3) = c_{12}(u_1,u_2) c_{23}(u_2,u_3) c_{13;2}(C_{1|2}(u_1|u_2),C_{3|2}(u_3|u_2))$.

\end{enumerate}
\end{thm}

Theorem \ref{thm:copuladeriv} shows that the copula density corresponding to the three-dimensional vector of block-maxima based on an arbitrary vine copula is numerically tractable since only one-dimensional integration is needed. In particular this allows a numerical treatment for the block-size $n$ in a finite setting. Additionally we can use the vine decomposition for a three-dimensional Gaussian or t-copula instead of requiring three-dimensional integration to calculate the corresponding density of the block-maxima. Two examples illustrate this way of proceeding.

\begin{exa}
\label{ex:clayton3d}
As a first example we take a three-dimensional Clayton-vine, i.e. all three pair-copulas are bivariate Clayton copulas. As parameters we choose $\delta_{12}=6$, $\delta_{23}=7.09$ and $\delta_{13;2}=4.67$ corresponding to Kendall's $\tau$ values of $\tau_{12}=0.75$ and $\tau_{23}=0.78$, $\tau_{13;2}=0.70$. 
\begin{figure}[!htb]
	\centering
		\includegraphics[trim=0cm 7cm 0cm 0.2cm,clip,width=\textwidth]{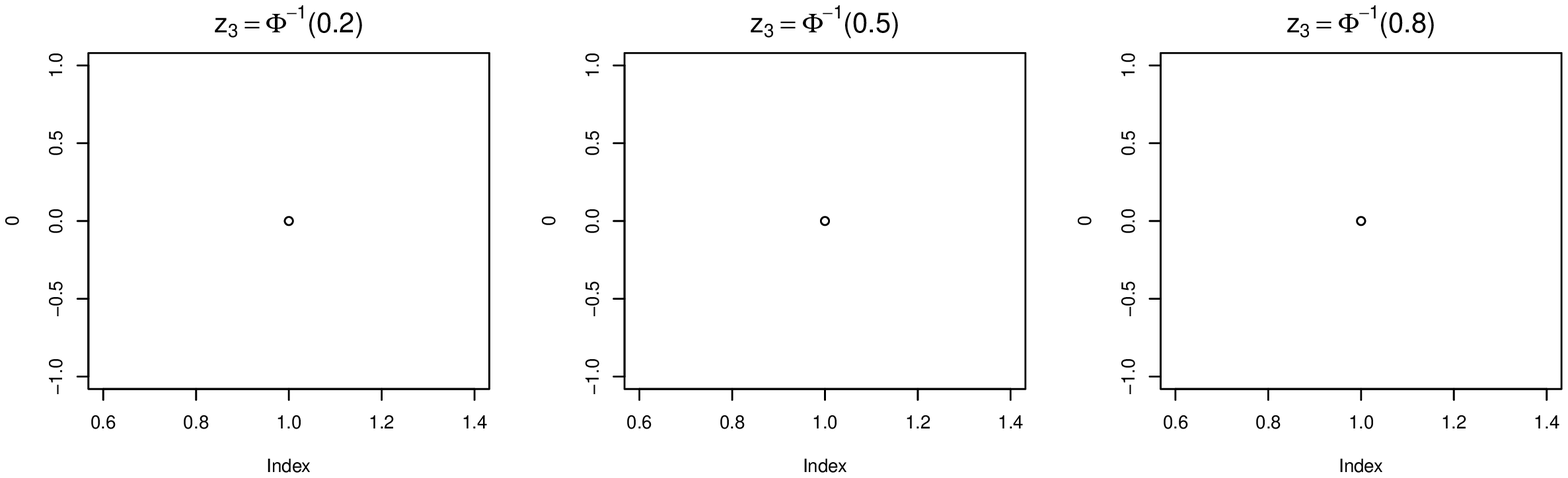}
		\includegraphics[width=\textwidth]{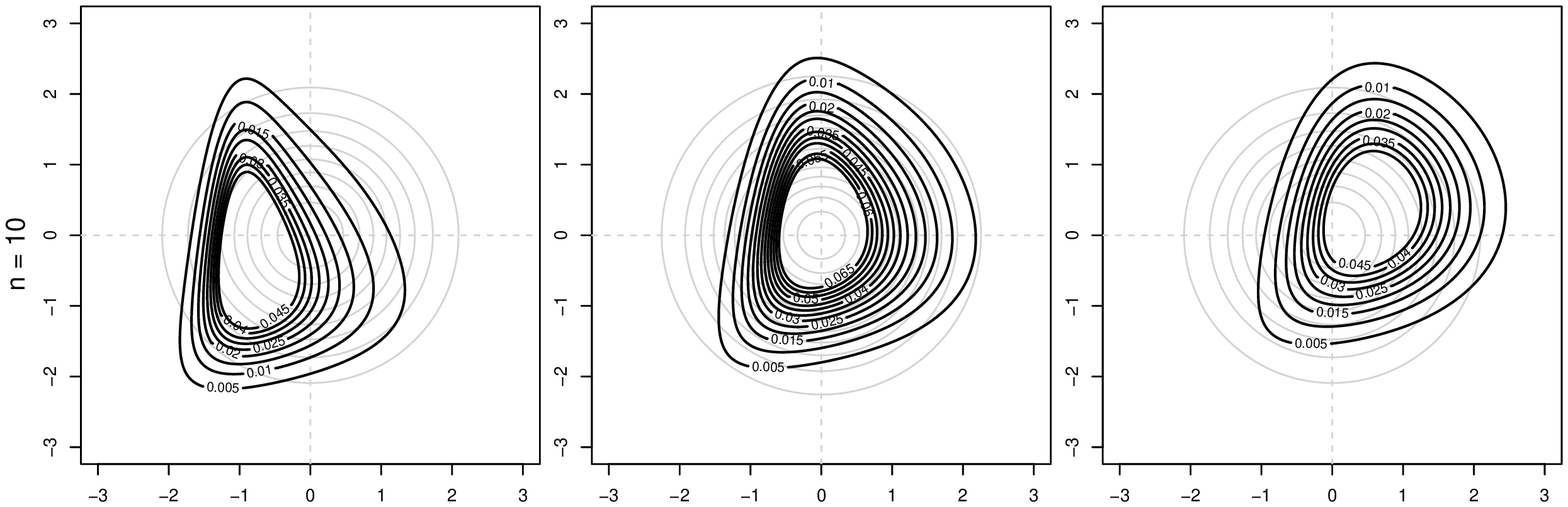}
		\includegraphics[width=\textwidth]{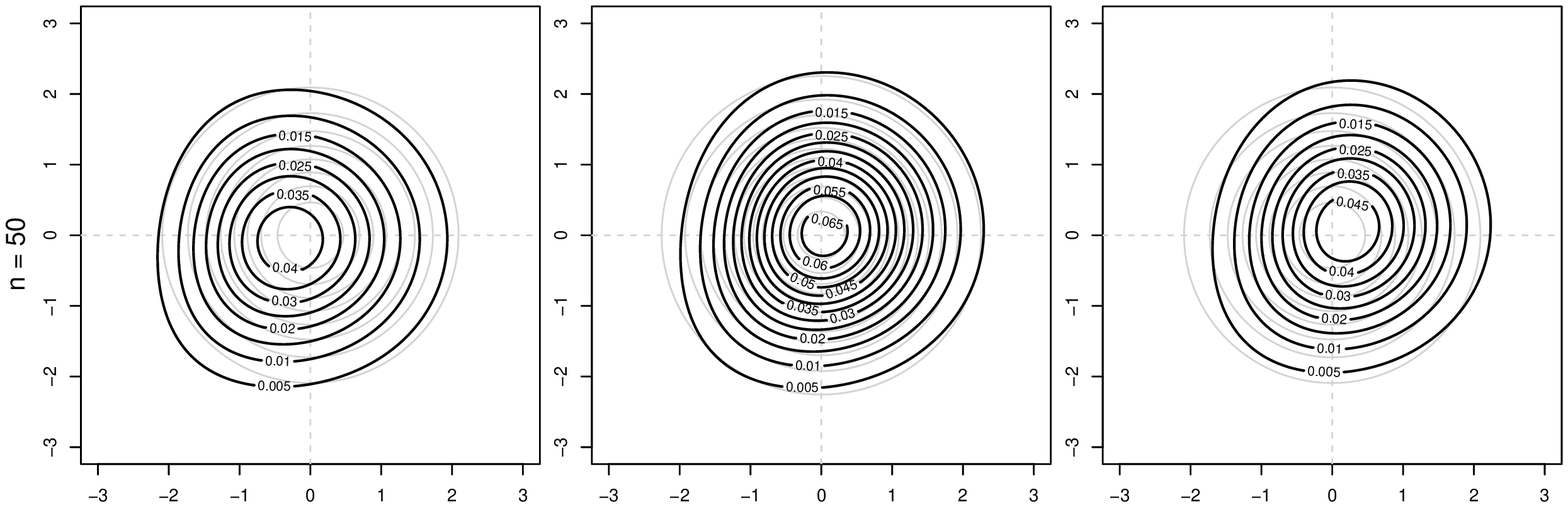}
		\includegraphics[width=\textwidth]{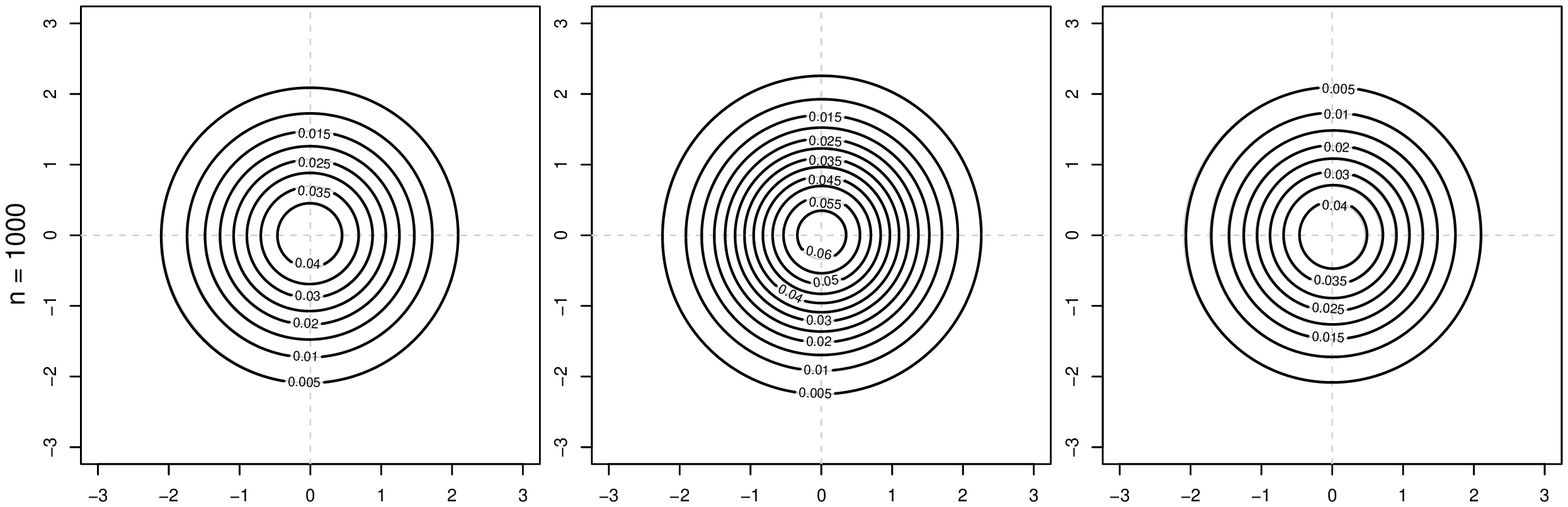}
	\caption{Two-dimensional slices of the copula density for block-maxima of a three-dimensional Clayton-vine with normalized margins ($\tau_{12}=0.75$, $\tau_{23}=0.78$, $\tau_{13;2}=0.7$).}
	\label{fig:ClaytonUnscaled}
\end{figure}
Figure \ref{fig:ClaytonUnscaled} shows the copula density of the block-maxima of this vine with normalized margins (i.e. on the z-level) for block-sizes $n=10,50,10^3$. Each row represents one block-size and contains three contour plots. Since it is difficult to plot three-dimensional objects in a simple way we decided not to show the isosurfaces but cut the three-dimensional object into three slices parallel to the z$_{1}$-z$_{2}$-plane. Each column presents the contourplot of one slice, where the z$_{3}$-value is fixed to $\Phi^{-1}(0.2)$, $\Phi^{-1}(0.5)$ or $\Phi^{-1}(0.8)$, respectively. Furthermore, we plotted the contours of the independence copula with normalized margins. One can see that already for $n=10^3$ the contours of the copula density of the block-maxima with normalized margins practically coincide with the ones of the independence copula.
\end{exa}

\begin{rem}
Even though it is not known whether all Clayton-vines lie in the domain of attraction of the independence copula, one can show that the Clayton-copula, which can be represented as a Clayton-vine with specific parameter restrictions (\cite{stoeber-vines}), lies in the domain of attraction of the independence copula. According to \cite{gudendorf2010extreme} an Archimedean copula with generator $\varphi$ lies in the domain of attraction of the Gumbel-copula with parameter $\theta:=-\lim_{s\downarrow 0}\frac{s\varphi'(1-s)}{\varphi(1-s)}\in[1,\infty)$ if the limit exists. For the Clayton-copula this limit is equal to 1. Therefore, the copula of the block-maxima of a Clayton-copula converges to the Gumbel-copula with $\theta=1$, which is the independence copula.
\end{rem}
\begin{figure}[!htb]
	\centering
		\includegraphics[trim=0cm 7cm 0cm 0.2cm,clip,width=\textwidth]{caption.eps}
		\includegraphics[width=\textwidth]{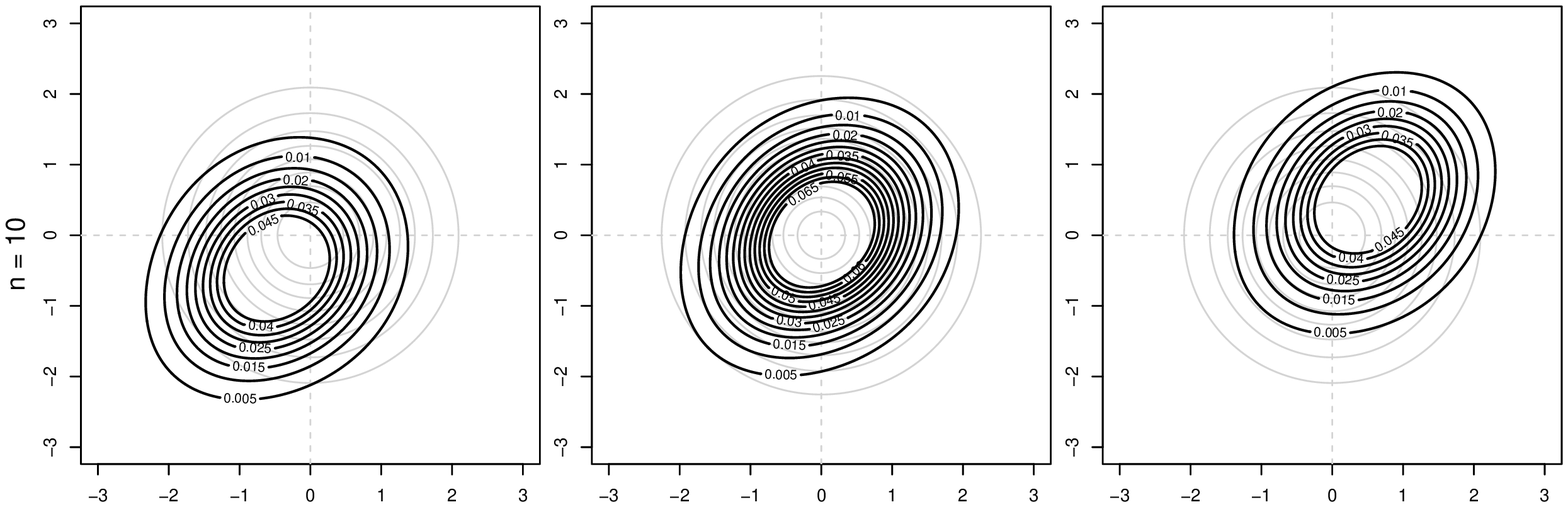}
		\includegraphics[width=\textwidth]{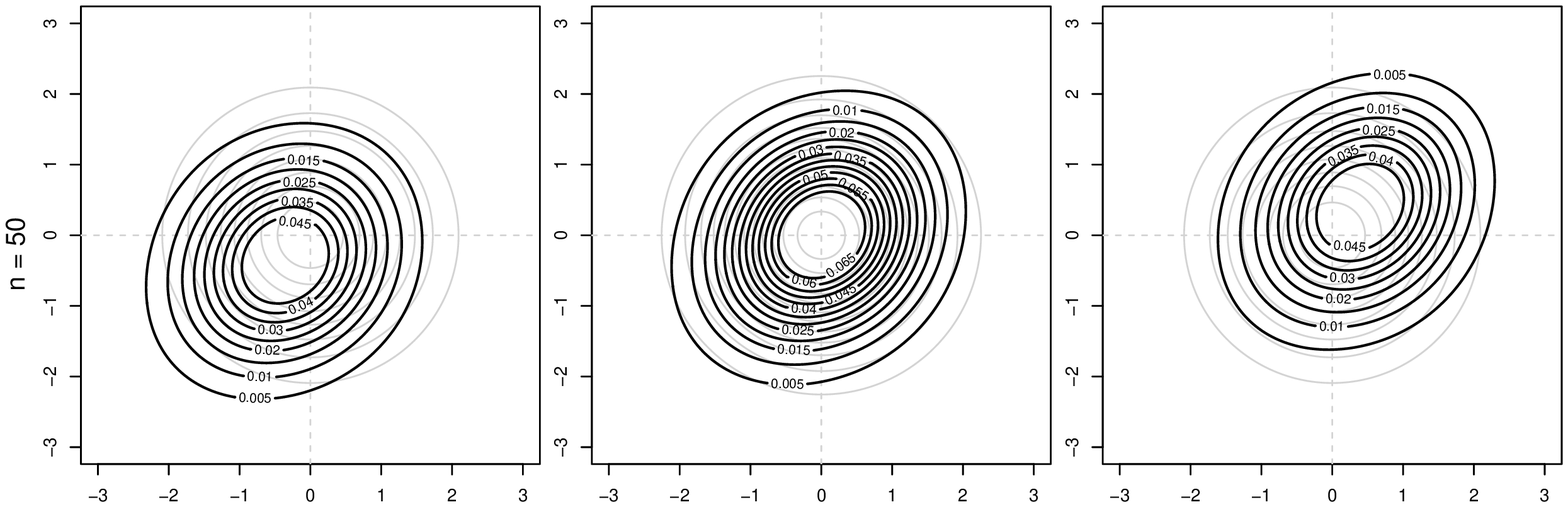}
		\includegraphics[width=\textwidth]{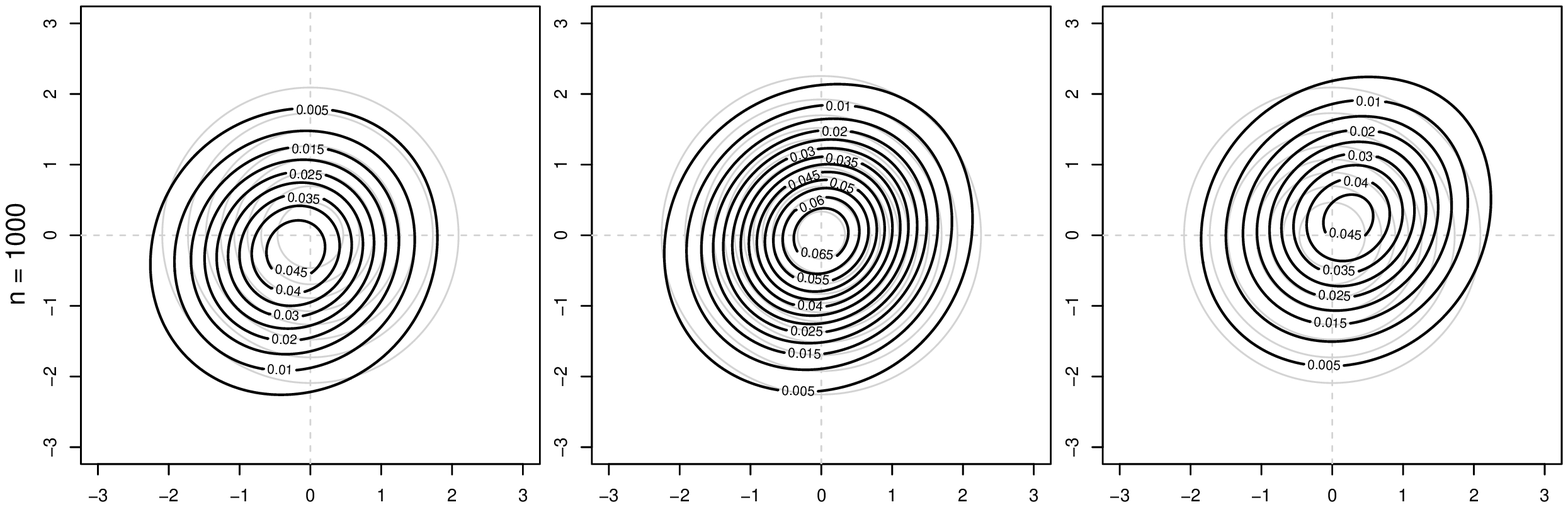}
	\caption{Two-dimensional slices of the copula density for block-maxima of a three-dimensional Gaussian vine with normalized margins ($\tau_{12}=0.5$, $\tau_{23}=0.57$, $\tau_{13;2}=0.35$).}
	\label{fig:GaussUnscaled}
\end{figure}

\begin{exa}
\label{ex:gauss3d}
The second example we present is a three-dimensional Gaussian vine, i.e. all three pair-copulas are bivariate Gaussian copulas. As parameters we choose $\rho_{12}=0.71$, $\rho_{23}=0.78$ and $\rho_{13;2}=0.52$ corresponding to Kendall's $\tau$ values of $\tau_{12}=0.50$ and $\tau_{23}=0.57$, $\tau_{13;2}=0.35$. 

Figure \ref{fig:GaussUnscaled} shows the copula density of the block-maxima of this vine with normalized margins (i.e. on the z-level) for block-sizes $n=10,50,10^3$. As above each row represents one block-size and contains three contour plots corresponding to z$_{3}$-values fixed to $\Phi^{-1}(0.2)$, $\Phi^{-1}(0.5)$ or $\Phi^{-1}(0.8)$, respectively. Again we detect convergence to the independence copula. This is also what one would expect: \cite{husler1989maxima} showed that in order to achieve that the distribution of the maxima of a multivariate Gaussian distribution converges to a non-trivial limiting distribution, a proper scaling of the margins and the correlation coefficients is necessary. This will be discussed in Section \ref{sec:scaledmax}.
\end{exa}
\begin{exa}\label{exa:data}
Hydrology is one of the areas where block-maxima are important. Especially, the water levels of rivers can be interesting when it comes to analyzing the risk of floods. We consider a three-dimensional data set containing the water levels of rivers in and around Munich, Germany, from August 1, 2007 to July 31, 2013. The data has been taken from Bavarian Hydrological Service (\texttt{http://www.gkd.bayern.de}). The three variables denote the differences of the 12 hour average water levels at the following three measuring points: the Isar measured in Munich, the Isar measured in Baierbrunn (south of Munich) and the Schwabinger Bach measured in Munich (a small stream entering the Isar in Garching, north of Munich). Since we only consider the hydrological winter (November 1 to April 30), we have 2176 data points.

First, we transform the margins to the unit interval applying the probability integral transform with the empirical marginal distribution functions. Then, we estimate the dependence structure using vine copulas\footnote{In order to assure that the necessary integrals are numerically tractable we had to exclude some pair-copula families (e.g. the t-copula).}: $c_{12}$ is estimated to be a Frank copula with a Kendall's $\tau$ of $\tau_{12}=0.76$, $c_{23}$ is a Frank copula with $\tau_{23}=0.23$ and $c_{13;2}$ is a Gaussian copula with $\tau_{13;2}=-0.18$. Now we are interested in the resulting copula density of the maxima for one day ($n=2$), one week ($n=14$), one month ($n=60$) and one winter ($n=362$). The respective contours (on the z-scale) are plotted in Figure \ref{fig:WD}.
\begin{figure}[!htb]
	\centering
		\includegraphics[trim=0cm 7cm 0cm 0.2cm,clip,width=\textwidth]{caption.eps}
		\includegraphics[width=\textwidth]{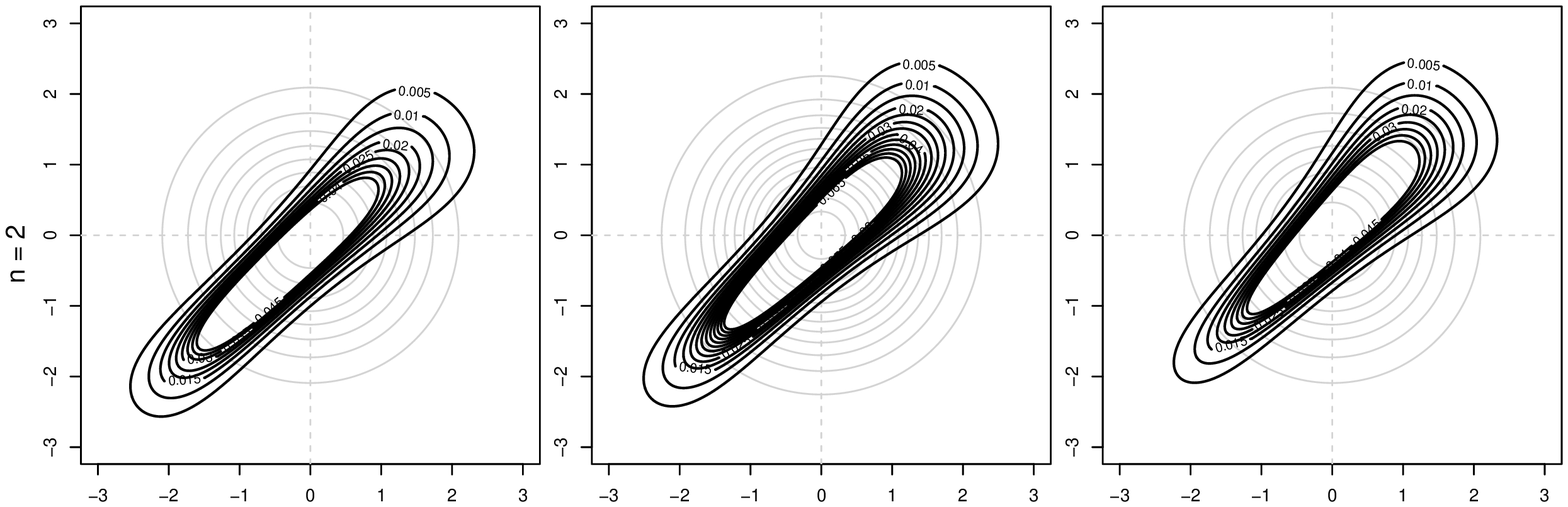}
		\includegraphics[width=\textwidth]{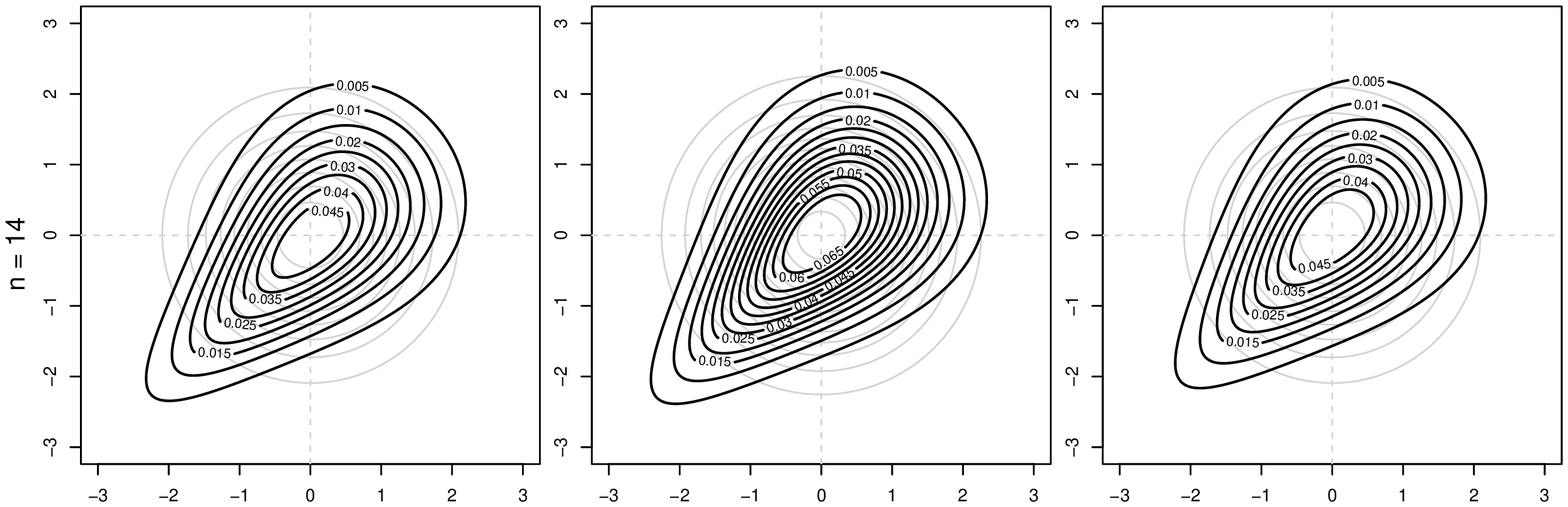}
		\includegraphics[width=\textwidth]{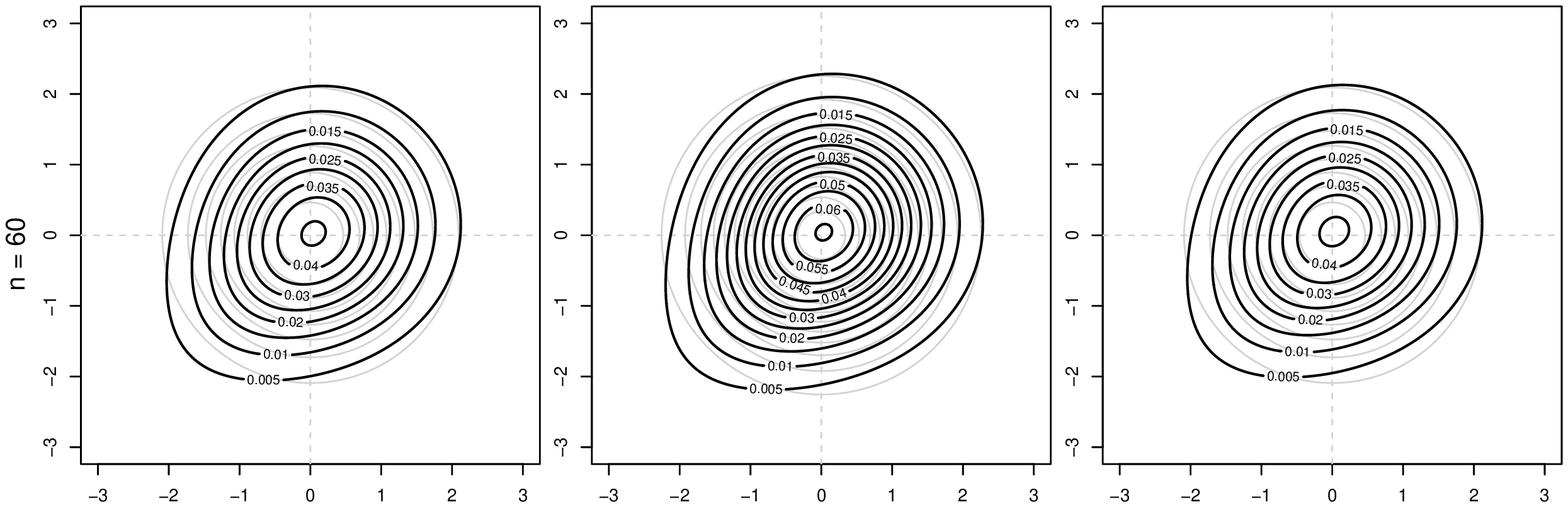}
		\includegraphics[width=\textwidth]{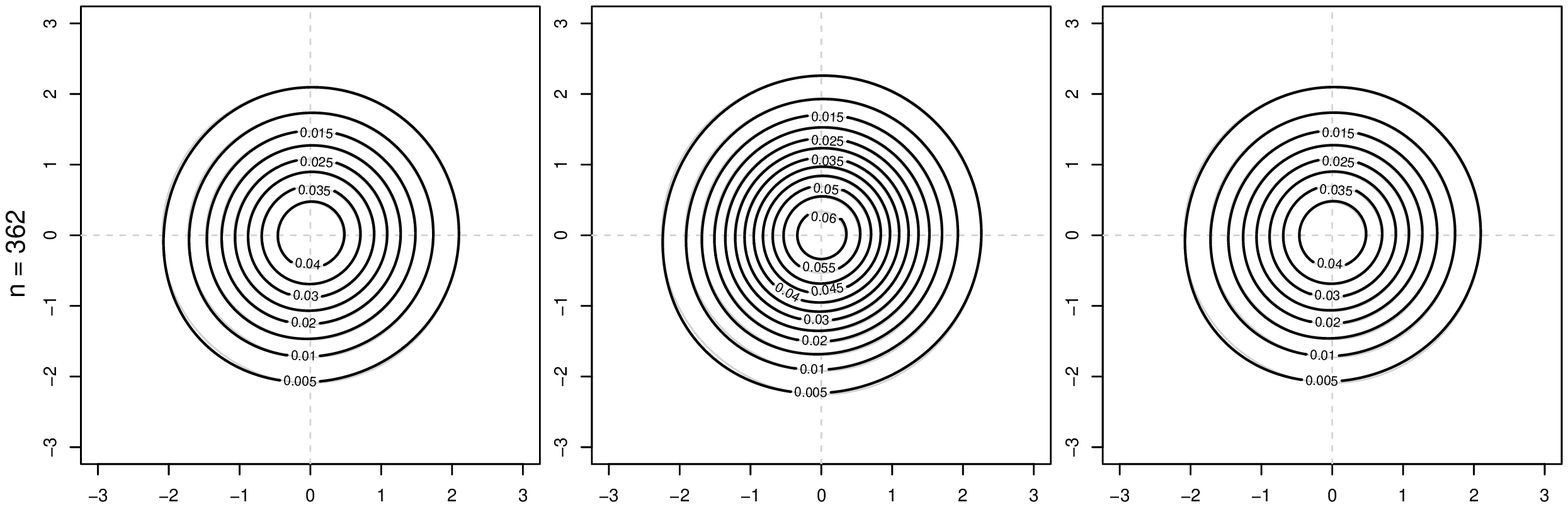}
	\caption{Two-dimensional slices of the copula density for block-maxima of the water level differences for one day, one week, one month and one winter with normalized margins.}
	\label{fig:WD}
\end{figure}
Similar to the examples from above we see that with increasing $n$ the observed dependence structure tends to the independence copula (gray contours in the background). In case of the considered rivers this means that the maximal differences of the 12 hour average water levels over the entire winter are almost independent.
\end{exa}

\section{Copula Density of Scaled Block-Maxima}
\label{sec:scaledmax}

Examples \ref{ex:clayton3d} and \ref{ex:gauss3d} show that scaling of the block-maxima is necessary to achieve a limiting copula. These limiting copulas are called extreme value copulas and are characterized by max-stability. A recent introduction to extreme value copulas is given by \cite{gudendorf2010extreme}.

Since \cite{husler1989maxima} derived the scaling for the block-maxima of the multivariate normal distribution with standard normally distributed margins $X_1,\ldots,X_d$ to a non-trivial extreme value copula, we 
use the same marginal scaling for the block-maxima $\Mb^{(n)}$ on the z-scale given by

\[
W^{(n)}_j := b_n \left(M_j^{(n)}-b_n\right),
\]
where $b_n$ satisfies $b_n=n\cdot\varphi(b_n)$ for $\varphi$ the standard normal density. Univariate extreme value theory gives that
\[
F_{W^{(n)}_j}(w_j)=P\left(W^{(n)}_j \leq w_j\right)=\Phi^n\left(b_n+ \frac{w_j}{b_n}\right) \to
\exp\{-e^{-w_j}\} \mbox{ as } n \to \infty
\]
for $w_j\in\R$. The marginal density of $W_j^{(n)}$ is given by
\[
f_{W^{(n)}_j}(w_j)=\frac{n}{b_n}\Phi^{n-1}\left(b_n+\frac{w_j}{b_n}\right)\varphi\left(b_n+\frac{w_j}{b_n}\right)
\]
for $w_j\in\R$, $j=1,\ldots,d$. Since $W_j^{(n)}$ is a strictly increasing transformation of $M_j^{(n)}$, the copula of $\Wb^{(n)}:=\left(W_1^{(n)},\ldots,W_d^{(n)}\right)$ is the same as the one of $\Mb^{(n)}$. Therefore, using \eqref{eq:maxcop} we obtain the following expression for the joint distribution of $\Wb^{(n)}$ 
\begin{equation*}
\begin{split}
F_{\Wb^{(n)}}(w_1,\ldots,w_d)&=P\left(W^{(n)}_1 \leq w_1, \ldots, W^{(n)}_d \leq w_d\right)\\
&=C_{\Mb^{(n)}}\left(\Phi^n\left(b_n+ \frac{w_1}{b_n}\right), \ldots,
\Phi^n\left(b_n+ \frac{w_d}{b_n}\right)\right) \\
 &=\left[ C\left(\Phi\left(b_n+ \frac{w_1}{b_n}\right), \ldots,
\Phi\left(b_n+ \frac{w_d}{b_n}\right)\right) \right]^n.
\end{split}
\end{equation*}
Similar arguments as in Corollary \ref{cor:maxzden} can be used to express the joint density of $\Wb^{(n)}$ in three dimensions for $n\geq3$ as
\[
\begin{split}
f&_{\Wb^{(n)}}(w_1,w_2,w_3)=\frac{1}{b_n^3}\prod_{j=1}^d \varphi\left(b_n+\frac{w_j}{b_n}\right)\Big\{nC\left(u_1,u_2,u_3\right)^{n-1}c\left(u_1,u_2,u_3\right)\Big.\\
&+n(n-1)C\left(u_1,u_2,u_3\right)^{n-2}\cdot\Big[\partial_1C\left(u_1,u_2,u_3\right)\partial_{23}C\left(u_1,u_2,u_3\right)\Big.
\end{split}
\]
\begin{equation}\label{eq:fWn}
\begin{split}
\quad&+\partial_2C\left(u_1,u_2,u_3\right)\partial_{13}C\left(u_1,u_2,u_3\right)\Big.+\partial_3C\left(u_1,u_2,u_3\right)\partial_{12}C\left(u_1,u_2,u_3\right)\Big]\\
\quad&+n(n-1)(n-2)C\left(u_1,u_2,u_3\right)^{n-3}\partial_1C\left(u_1,u_2,u_3\right)\\
\quad&\Big.\cdot\partial_2C\left(u_1,u_2,u_3\right)\partial_3C\left(u_1,u_2,u_3\right)
\Big\},
\end{split}
\end{equation}
where $u_j:=\Phi\left(b_n+ \frac{w_j}{b_n}\right)$ for $j=1,2,3$.\\

According to \cite{husler1989maxima}, besides scaling the maxima $M_1^{(n)},\ldots,M_d^{(n)}$, it is also necessary to change the correlation matrix $\Sigma(n)=\left(\rho_{i,j}(n)\right)_{1\leq i,j \leq d}$ of the underlying joint distribution of standard normal random variables $X_1,\ldots,X_n$, over whose i.i.d. copies $X_{i,1},\ldots,X_{i,d}$, $i=1,\ldots,n$, we take the maximum. The correlation matrices $\Sigma(n)$ have to satisfy the following condition
\begin{equation}\label{eq:rholambda}
\left(1-\rho_{i,j}(n)\right)\cdot\log(n)\to\lambda_{i,j}^2 \mbox{ as } n \to \infty,
\end{equation}
where $\lambda_{i,j}\in(0,\infty)$ are some constants for $1\leq i,j\leq d$, $i\neq j$ and $\lambda_{i,i}=0$ for $i=1,\ldots,d$. Since $\rho_{i,j}(n)=\rho_{j,i}(n)$, we also have $\lambda_{i,j}=\lambda_{j,i}$ for $1\leq i, j \leq d$. Note that \eqref{eq:rholambda} implies that $\rho_{i,j}(n)\to 1$ as $n \to \infty$. The limiting distribution $H_{\Lambda}$ of the scaled maxima depends on $\Lambda:=\left(\lambda_{i,j}\right)_{1\leq i,j \leq d}$.\\

In the following we will examine the three-dimensional case and choose values for $\lambda_{12},\lambda_{13},\lambda_{23}\in(0,\infty)$. For simplicity, we set
\begin{equation}\label{eq:rhon}
\rho_{i,j}(n):=1-\frac{\lambda_{i,j}^2}{\log(n)}
\end{equation}
for $1\leq i,j\leq 3$, $n\in\N$, such that Equation \ref{eq:rholambda} is always satisfied. However, for arbitrary $\lambda_{i,j}$ it is not trivial to decide whether we obtain a valid correlation matrix through this particular choice of $\rho_{i,j}(n)$ for any $n\in\N$. By construction the matrices 
\[
\Sigma(n)=\begin{pmatrix} 1 & \rho_{12}(n)& \rho_{13}(n)\\ 
												 \rho_{12}(n) & 1 & \rho_{23}(n)\\
												 \rho_{13}(n) & \rho_{23}(n) & 1\end{pmatrix}
\]
are symmetric and have ones on their diagonals. The only property we have to check is whether $\Sigma(n)$ is positive definite. For this we only need to check if the determinant of each leading principal minor is positive. Since $1>0$ and $1-\rho_{12}(n)^2>0$ are trivially satisfied, the only real requirement is that 
\[
\det(\Sigma(n))=1-\rho_{12}(n)\rho_{13}(n)\rho_{23}(n)-\rho_{12}(n)^2-\rho_{13}(n)^2-\rho_{23}(n)^2>0.
\]
Using Equation \ref{eq:rhon} we obtain that $\det(\Sigma(n))>0$ if and only if
\begin{equation}\label{eq:h}
2(\lambda_{12}^2\lambda_{13}^2+\lambda_{12}^2\lambda_{23}^2+\lambda_{13}^2\lambda_{23}^2)-(\lambda_{12}^4+\lambda_{13}^4+\lambda_{23}^4) \\>\frac{2\lambda_{12}^2\lambda_{13}^2\lambda_{23}^2}{\log(n)}.
\end{equation}
We denote the left-hand side of Equation \ref{eq:h} by $h(\lambda_{12}^2,\lambda_{13}^2,\lambda_{23}^2)$. Since the right-hand side of Equation \ref{eq:h} is always positive, it can only be satisfied if $h(\lambda_{12}^2,\lambda_{13}^2,\lambda_{23}^2)>0$. If $h(\lambda_{12}^2,\lambda_{13}^2,\lambda_{23}^2)>0$, then \eqref{eq:h} is satisfied for all $n\in\N$ with
\[
n\geq n^*:=\left\lfloor \exp\left\{\frac{2\lambda_{12}^2\lambda_{13}^2\lambda_{23}^2}{h(\lambda_{12}^2,\lambda_{13}^2,\lambda_{23}^2)}\right\}+1\right\rfloor,
\]
where $\left\lfloor \cdot\right\rfloor$ denotes the floor function. Table \ref{tab:hn} shows the values of $h$ and $n^*$ (if existing) for 10 different combinations of $\lambda_{12},\lambda_{13},\lambda_{23}$.

\begin{table}[ht]
\centering
\begin{tabular}{rrrrrr}
  \hline
 \rule{0pt}{2.5ex} \# & $\lambda_{12}^2$ & $\lambda_{13}^2$ & $\lambda_{23}^2$ & $h(\lambda_{12}^2,\lambda_{13}^2,\lambda_{23}^2)$ & $n^*$ \\ 
  \hline
	1 & 1 & 1 & 1 & 3 & 2 \\ 
  2 & 2 & 2 & 2 & 12 & 4 \\ 
  3 & 1 & 2 & 3 & 8 & 5 \\ 
  4 & 0.5 & 0.5 & 0.5 & 0.75 & 2 \\ 
  5 & 0.3 & 0.2 & 0.1 & 0.08 & 2 \\ 
  6 & 0.2 & 5 & 0.75 & -15.8 & -- \\ 
  7 & 15 & 20 & 15 & 800 & 76880 \\ 
  8 & 100 & 0.1 & 20 & -6376.01 & -- \\ 
  9 & 1.05 & 0.21 & 0.84 & 0.71 & 2 \\ 
  10 & 4 & 3 & 3 & 32 & 10\\ 
   \hline
\end{tabular}
\caption{Different $\lambda$-combinations and the corresponding values of $h$ and $n^*$.}
\label{tab:hn}
\end{table}
\cite{husler1989maxima} derived this scaling for multivariate normal distributions. Since we want to apply the scaling to vines, we need to transform the parameters of the normal distribution (correlations) to the parameters of the vine.

Considering the vine structure from \eqref{eq:pcc} we further assume that the pair-copulas are one-parametric. Having fixed $\lambda_{12},\lambda_{13},\lambda_{23}$ such that $h(\lambda_{12}^2,\lambda_{13}^2,\lambda_{23}^2)>0$, we can perform the following procedure for $n\geq n^*$:
\begin{enumerate}
	\item Calculate $\rho_{12}(n)$, $\rho_{13}(n)$ and $\rho_{23}(n)$ with the help of \eqref{eq:rhon}.
	\item Determine the corresponding partial correlation via
	\[	
	\rho_{13;2}(n)=\frac{\rho_{13}(n)-\rho_{12}(n)\rho_{23}(n)}{\sqrt{1-\rho_{12}(n)^2}\sqrt{1-\rho_{23}(n)^2}}.
	\]
	\item Translate the (partial) correlations $\rho_{12}(n)$, $\rho_{23}(n)$ and $\rho_{13;2}(n)$ into (partial) Kendall's $\tau$ values $\tau_{12}(n)$, $\tau_{23}(n)$ and $\tau_{13;2}(n)$ using the relation for elliptical distributions
	\[
	\tau=\frac{2}{\pi}\arcsin(\rho).
	\] 
	\item Determine the parameters $\theta_{12}(n)$, $\theta_{13}(n)$ and $\theta_{23}(n)$ of the pair copulas from the corresponding $\tau$ values.\footnote{In the \texttt{VineCopula} package this transformation can be performed by the function \texttt{BiCopTau2Par}.}
\end{enumerate}

Recall that $\rho_{12}(n)\to 1$, $\rho_{13}(n)\to 1$ and $\rho_{23}(n)\to 1$ as $n\to \infty$. Therefore, we also have $\tau_{12}(n)\to 1$, $\tau_{13}(n)\to 1$ and $\tau_{23}(n)\to 1$ as $n\to \infty$. However, the behavior of convergence of $\rho_{13;2}(n)$ and hence $\tau_{13;2}(n)$ is not trivial. We use \eqref{eq:rhon} to obtain
\[
\begin{split}
\rho_{13;2}(n)&=\frac{\left(1-\frac{\lambda_{13}^2}{\log(n)}\right)-\left(1-\frac{\lambda_{12}^2}{\log(n)}\right)\left(1-\frac{\lambda_{23}^2}{\log(n)}\right)}{\sqrt{1-\left(1-\frac{\lambda_{12}^2}{\log(n)}\right)^2}\sqrt{1-\left(1-\frac{\lambda_{23}^2}{\log(n)}\right)^2}}
\to \frac{\lambda_{12}^2+\lambda_{23}^2-\lambda_{13}^2}{2\lambda_{12}\lambda_{23}}
\end{split}
\]
as $n\to\infty$. Thus, 
\[
\tau_{13;2}(n)\to \frac{2}{\pi}\arcsin\left( \frac{\lambda_{12}^2+\lambda_{23}^2-\lambda_{13}^2}{2\lambda_{12}\lambda_{23}}\right) \mbox{ as } n\to \infty. \\
\]
For illustration, we will now take combinations 9 and 10 from Table \ref{tab:hn}. We show the (partial) correlations from Step 2 of the above procedure as well as the (partial) Kendall's $\tau$ values since they can be compared independently from the choice the respective pair-copulas.
\begin{table}[ht]
\centering
\begin{tabular}{l|rrrrrrr}
  \hline
 \rule{0pt}{2ex}	&$n$ & $\rho_{12}(n)$ & $\rho_{23}(n)$& $\rho_{13;2}(n)$ & $\tau_{12}(n)$ & $\tau_{23}(n)$& $\tau_{13;2}(n)$ \\ 
  \hline
	 \rule{0pt}{2ex}	Combination 9 &10 & 0.54 & 0.64 & 0.87 & 0.37 & 0.44 & 0.67 \\ 
   &50 & 0.73 & 0.79 & 0.88 & 0.52 & 0.57 & 0.69 \\ 
   &$10^3$ & 0.85 & 0.88 & 0.89 & 0.64 & 0.68 & 0.69 \\ 
   &$\infty$ & 1 & 1 & 0.89 & 1 & 1 & 0.70 \\
	\hline
	\rule{0pt}{2ex} Combination 10 &10 & -0.74 & -0.30 & -0.82 & -0.53 & -0.20 & -0.61 \\ 
  &50 & -0.02 & 0.23 & 0.25 & -0.01 & 0.15 & 0.16 \\ 
  &$10^3$ & 0.42 & 0.57 & 0.44 & 0.28 & 0.38 & 0.29 \\ 
	&$\infty$ & 1 & 1 & 0.58 & 1 & 1 & 0.39 \\
	 \hline
\end{tabular}
\caption{Overview over the (partial) correlations and (partial) Kendall's $\tau$ values for different $n$ for combinations 9 ($\lambda_{12}^2=1.05$, $\lambda_{13}^2=0.21$, $\lambda_{23}^2=0.84$) and 10 ($\lambda_{12}^2=4$, $\lambda_{13}^2=3$, $\lambda_{23}^2=3$).}
\label{tab:rhotau1}
\end{table}

If we compare the values from Table \ref{tab:rhotau1} for combinations 9 and 10, it is eye-catching that the choice of $\lambda_{12}$, $\lambda_{13}$ and $\lambda_{23}$ has a crucial influence on the behavior of the (partial) correlations and the (partial) Kendall's $\tau$ values. In the first case the parameters are already relatively close to their limiting values for $n=10^3$, whereas in the second case they are still rather far from their limits for $n=10^3$. Further, we see that the limiting values of $\rho_{13;2}(n)$ and $\tau_{13;2}(n)$ can be very different depending on the choice $\lambda_{12}$, $\lambda_{13}$ and $\lambda_{23}$.\\

Now we examine the behavior of the three-dimensional density of the scaled block-maxima $f_{\Wb^{(n)}}$ for increasing values of $n$.
\begin{figure}[!htb]
	\centering
		\includegraphics[trim=0cm 7cm 0cm 0.2cm,clip,width=\textwidth]{caption.eps}
		\includegraphics[width=\textwidth]{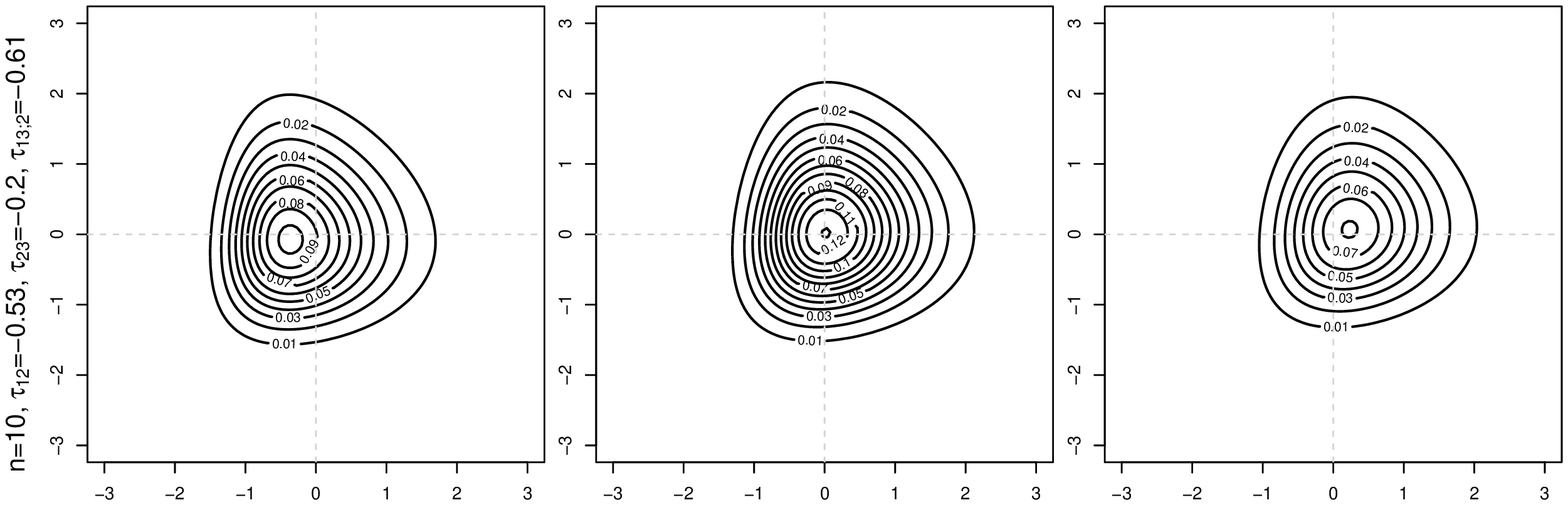}
		\includegraphics[width=\textwidth]{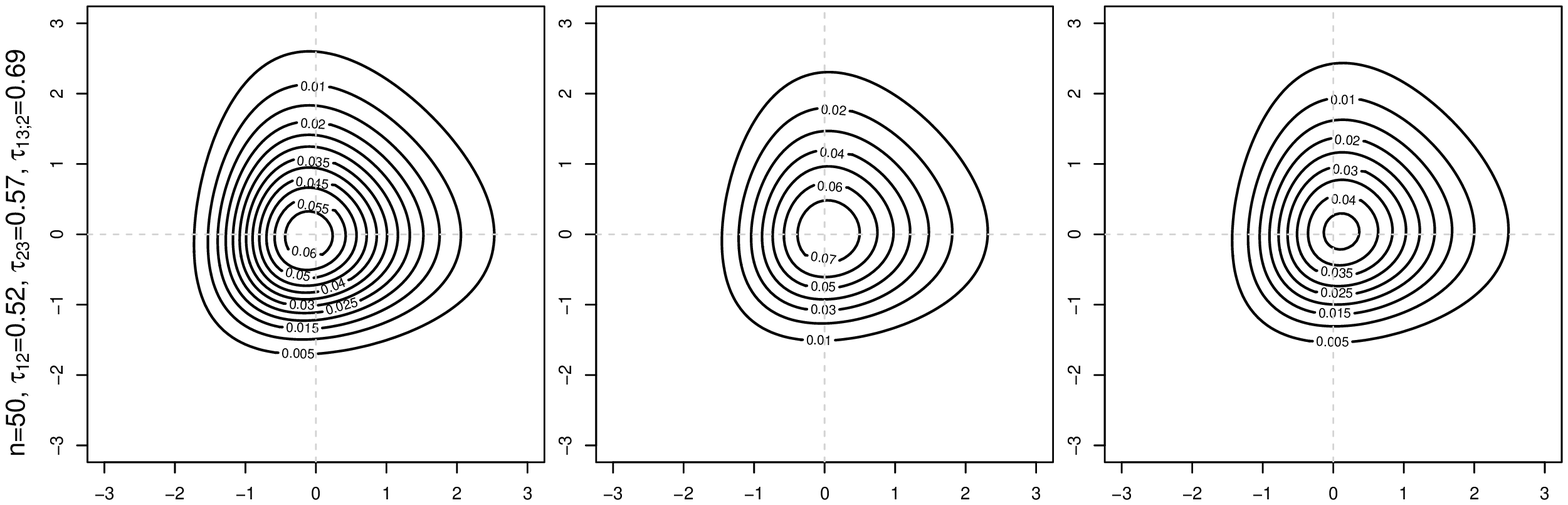}
		\includegraphics[width=\textwidth]{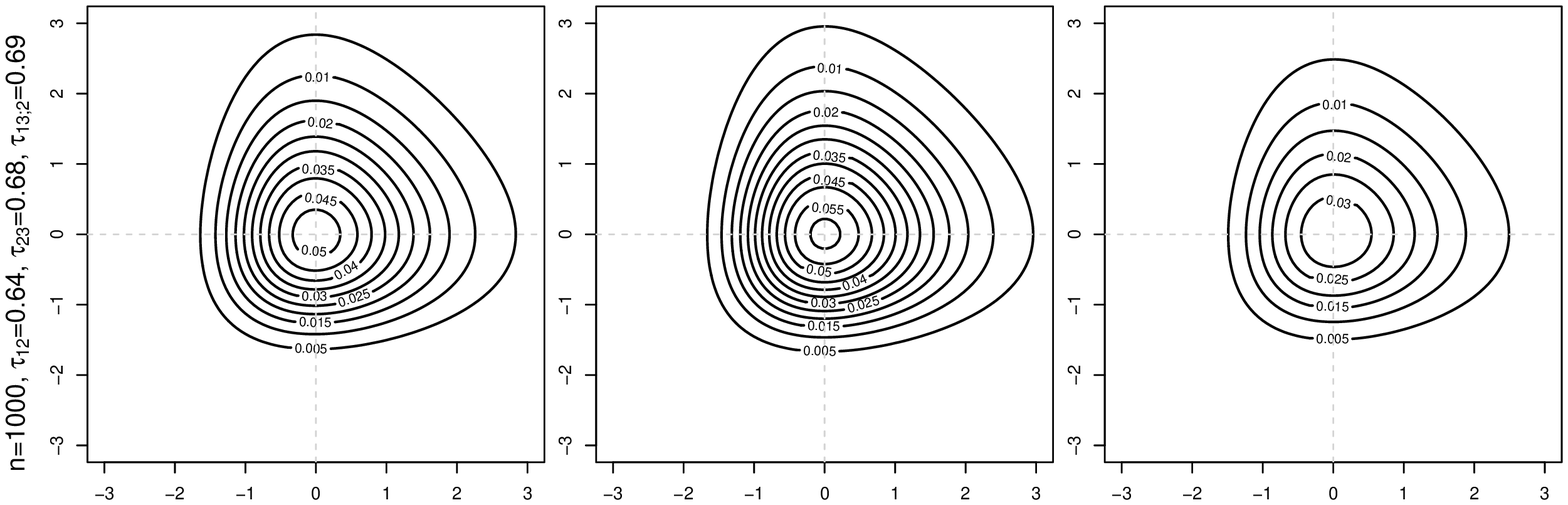}
	\caption{Two-dimensional slices of the density of the scaled block-maxima of a three-dimensional Clayton-vine ($\lambda_{12}^2=1.05$, $\lambda_{13}^2=0.21$, $\lambda_{23}^2=0.84$).}
	\label{fig:ClaytonScaled}
\end{figure}
\begin{exa}
\label{ex:clayton3dscaled}
First we look at a Clayton-vine and choose $\lambda_{12}^2=1.05$, $\lambda_{13}^2=0.21$ and $\lambda_{23}^2=0.84$ (combination 9). The parameters and Kendall's $\tau$ values depend on the block-size $n$.

Figure \ref{fig:ClaytonScaled} shows the density of the scaled block-maxima of the Clayton-vine for block-sizes $n=10$, $50$, $10^3$. Each row represents one block-size (and thus parameter set) and contains three contour plots corresponding to z$_{3}$-values fixed to $F_{W_3^{(n)}}^{-1}(0.2)$, $F_{W_3^{(n)}}^{-1}(0.5)$ or $F_{W_3^{(n)}}^{-1}(0.8)$, respectively. 
\end{exa}
\begin{figure}[!htb]
	\centering
		\includegraphics[trim=0cm 7cm 0cm 0.2cm,clip,width=\textwidth]{caption.eps}
		\includegraphics[width=\textwidth]{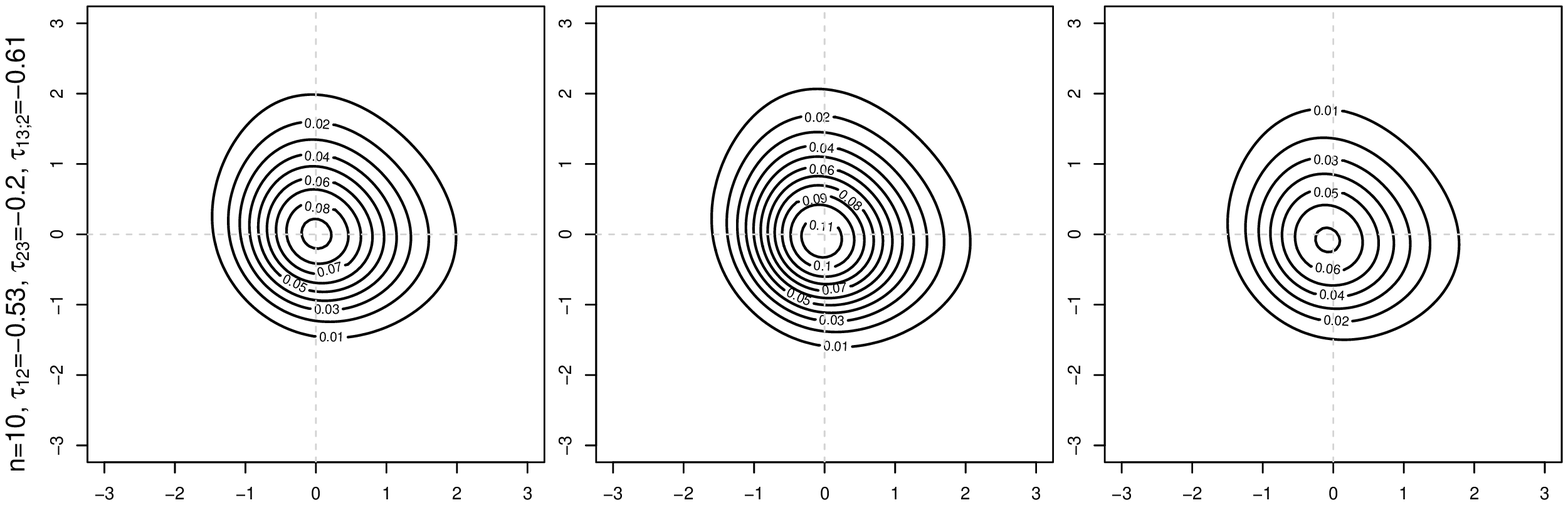}
		\includegraphics[width=\textwidth]{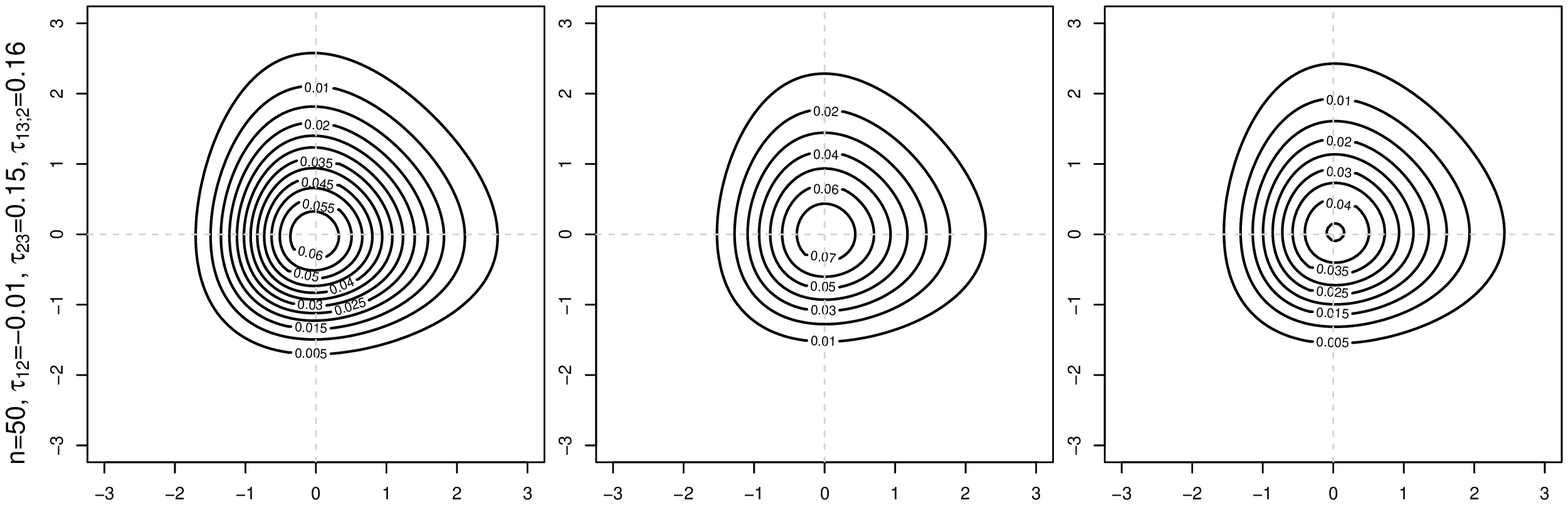}
		\includegraphics[width=\textwidth]{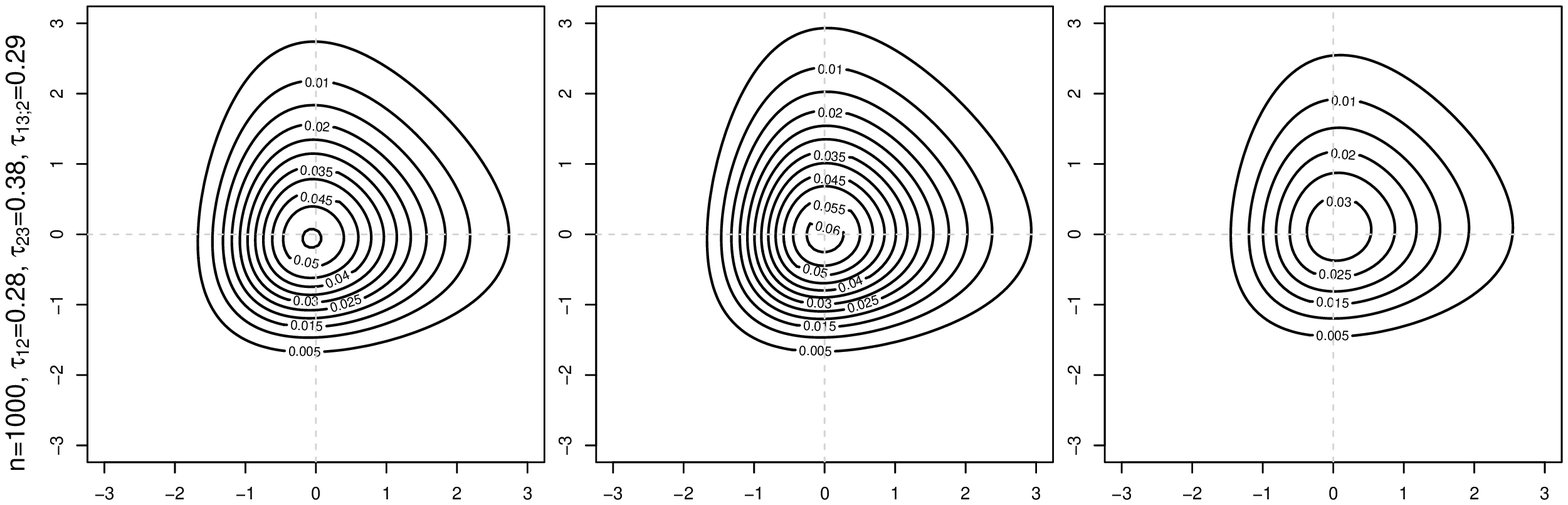}
	\caption{Two-dimensional slices of the density of the scaled block-maxima of a three-dimensional Gaussian vine ($\lambda_{12}^2=4$, $\lambda_{13}^2=3$, $\lambda_{23}^2=3$).}
	\label{fig:GaussScaled}
\end{figure}
\begin{exa}
\label{ex:gauss3dscaled}
As a second example we choose a Gaussian vine with $\lambda_{12}^2=4$, $\lambda_{13}^2=3$ and $\lambda_{23}^2=3$ (combination 10).

Figure \ref{fig:GaussScaled} shows the density of the scaled block-maxima of the Gaussian vine for block-sizes $n=10$, $50$, $10^3$. Again, three contour plots corresponding to z$_{3}$-values fixed to $F_{W_3^{(n)}}^{-1}(0.2)$, $F_{W_3^{(n)}}^{-1}(0.5)$ or $F_{W_3^{(n)}}^{-1}(0.8)$, respectively, are displayed per row. The block-size and Kendall's $\tau$ values are denoted on the left for each row. 
\end{exa}

\section*{Conclusion}
In this chapter we showed that the copula density of the block-maxima of multivariate distributions can be expressed explicitly. For three-dimensional vine copulas we made use of the fact that we can compute their partial derivatives by one-dimensional integration, which makes the evaluation of the copula density for block-maxima numerically tractable. The advantage of our method is that we can use the entire sample for estimation instead of reducing the sample size by taking the maximum over $n$ observations. Once we have estimated the underlying dependence structure we can derive the copula density of the block-maxima for any block-size (even larger than the original sample size). From the Clayton and the Gauss examples we have seen that without proper scaling the block-maxima do not approach a non-trivial limiting distribution for increasing block-size.

\section*{Appendix}\label{sec:appendix}
\subsection*{Proof of Theorem \ref{thm:copuladensity}}

In order to prove Theorem \ref{thm:copuladensity} we prove an auxiliary lemma from which Theorem \ref{thm:copuladensity} follows as a corollary.
\begin{lem}\label{lem:claim1}
For $k\in\left\{1,\ldots,d\right\}$ and $u_j\in[0,1]$, $j=1\ldots,d$, we have
\begin{multline*}
\frac{\partial^k }{\partial u_1\cdots \partial u_k}\left[C(u_1^{1/n},\ldots,u_d^{1/n})^n\right]=
\frac{1}{n^k}\left(\prod_{j=1}^k u_j\right)^{\frac{1}{n}-1}\\ \cdot\sum_{j=1}^{k \wedge n}\left\{\frac{n!}{(n-j)!}\cdot C\left(u_1^{1/n},\ldots,u_d^{1/n}\right)^{n-j} \sum_{\mathcal{P}\in\mathcal{S}_{k,j}} \prod_{M\in\mathcal{P}} \partial_M C\left(u_1^{1/n},\ldots,u_d^{1/n}\right) \right\}.
\end{multline*}

\end{lem}

\begin{proof}
We will prove this statement using induction. For $k=1$ we have
{\small
\begin{multline*}
\frac{\partial}{\partial u_1}\left[C(u_1^{1/n},\ldots,u_d^{1/n})^n\right]= n C(u_1^{1/n},\ldots,u_d^{1/n})^{n-1}\partial_1 C(u_1^{1/n},\ldots,u_d^{1/n})\frac{1}{n}u_1^{\frac{1}{n}-1} = \\
\frac{1}{n^1}\left(\prod_{j=1}^1u_j\right)^{\frac{1}{n}-1}\sum_{j=1}^{1\wedge n}\left\{ \frac{n!}{(n-j)!} C(u_1^{1/n},\ldots,u_d^{1/n})^{n-j}\sum_{\mathcal{P}\in\mathcal{S}_{1,j}}\prod_{M\in\mathcal{P}}\partial_M C(u_1^{1/n},\ldots,u_d^{1/n})\right\}.
\end{multline*}
}
\hspace{-1.75mm}The inductive step ($k\to k+1$) proceeds as follows
{\small
\[
\frac{\partial^{k+1} }{\partial u_1\cdots \partial u_{k+1}}\left[C(u_1^{1/n},\ldots,u_d^{1/n})^n\right]=\frac{\partial}{\partial u_{k+1}}\left\{ \frac{\partial^{k} }{\partial u_1\cdots \partial u_{k}}\left[C(u_1^{1/n},\ldots,u_d^{1/n})^n\right]\right\}
=:\left(*\right)_1.
\]
}
\hspace{-1.75mm}Applying the inductive assumption yields
\begin{multline*}
\left(*\right)_1=\frac{\partial}{\partial u_{k+1}}\left\{ \frac{1}{n^k}\left(\prod_{j=1}^k u_j\right)^{\frac{1}{n}-1} \sum_{j=1}^{k \wedge n}\left\{\frac{n!}{(n-j)!}\cdot C\left(u_1^{1/n},\ldots,u_d^{1/n}\right)^{n-j}\right.\right.\\
\cdot \big.\big.\sum_{\mathcal{P}\in\mathcal{S}_{k,j}} \prod_{M\in\mathcal{P}} \partial_M C\left(u_1^{1/n},\ldots,u_d^{1/n}\right) \bigg\}\Bigg\}=:\left(*\right)_2
\end{multline*}
We will consider the cases "$n>k$" and "$n\leq k$" separately. We begin with Case 1 ($n> k$): We have $k\wedge n=k$ and hence
\[
\begin{split}
\left(*\right)_2&=
\frac{1}{n^k}\left(\prod_{j=1}^k u_j\right)^{\frac{1}{n}-1} \sum_{j=1}^{k} \Bigg\{\frac{n!}{(n-j)!}\cdot \left\{\frac{\partial}{\partial u_{k+1}}\left[C\left(u_1^{1/n},\ldots,u_d^{1/n}\right)^{n-j}\right]\right.\Bigg.\\
&\quad\cdot \big.\sum_{\mathcal{P}\in\mathcal{S}_{k,j}} \prod_{M\in\mathcal{P}} \partial_M C\left(u_1^{1/n},\ldots,u_d^{1/n}\right) + C\left(u_1^{1/n},\ldots,u_d^{1/n}\right)^{n-j}\\
&\quad \cdot \frac{\partial}{\partial u_{k+1}}\left[\sum_{\mathcal{P}\in\mathcal{S}_{k,j}} \prod_{M\in\mathcal{P}} \partial_M C\left(u_1^{1/n},\ldots,u_d^{1/n}\right)\right]
\Bigg\}=\left(*\right)_3.
\end{split}
\]
Now we use that fact that for $k\in\left\{1,\ldots,d-1\right\}$ and $j\in\left\{1,\ldots,k\wedge n\right\}$ we have
{\small
 \be\label{eq:claim2}
\frac{\partial}{\partial u_{k+1}}\left[\sum_{\mathcal{P}\in\mathcal{S}_{k,j}} \prod_{M\in\mathcal{P}} \partial_M C\left(u_1^{1/n},\ldots,u_d^{1/n}\right)\right]= \frac{u_{k+1}^{\frac{1}{n}-1}}{n} \sum_{\substack{\mathcal{P}\in\mathcal{S}_{k+1,j} \\ \left\{k+1\right\}\not\in\mathcal{P}}} \prod_{M\in\mathcal{P}} \partial_M C\left(u_1^{1/n},\ldots,u_d^{1/n}\right).
\ee
}
\hspace{-1.75mm}Applying Equation \ref{eq:claim2} yields
{\small
\[
\begin{split}
\left(*\right)_3&=
\frac{1}{n^k}\left(\prod_{j=1}^k u_j\right)^{\frac{1}{n}-1} \Bigg\{ \sum_{j=1}^{k} \frac{n!}{(n-j)!}\cdot (n-j) \cdot C\left(u_1^{1/n},\ldots,u_d^{1/n}\right)^{n-j-1} \Bigg.\\
&\quad\cdot \partial_{k+1} C\left(u_1^{1/n},\ldots,u_d^{1/n}\right)\frac{1}{n}u_{k+1}^{\frac{1}{n}-1}\cdot \sum_{\mathcal{P}\in\mathcal{S}_{k,j}} \prod_{M\in\mathcal{P}} \partial_M C\left(u_1^{1/n},\ldots,u_d^{1/n}\right) \\
&\quad \big.+ \sum_{j=1}^{k} \frac{n!}{(n-j)!}\cdot C\left(u_1^{1/n},\ldots,u_d^{1/n}\right)^{n-j}\cdot \frac{1}{n}u_{k+1}^{\frac{1}{n}-1} \sum_{\substack{\mathcal{P}\in\mathcal{S}_{k+1,j} \\ \left\{k+1\right\}\not\in\mathcal{P}}} \prod_{M\in\mathcal{P}} \partial_M C\left(u_1^{1/n},\ldots,u_d^{1/n}\right)
\Bigg\}\\
&=\frac{1}{n^{k+1}}\left(\prod_{j=1}^{k+1} u_j\right)^{\frac{1}{n}-1} \Bigg\{ \sum_{j=1}^{k} \frac{n!}{(n-(j+1))!} \cdot C\left(u_1^{1/n},\ldots,u_d^{1/n}\right)^{n-(j+1)} \Bigg.\\
&\quad\cdot \sum_{\substack{\mathcal{P}\in\mathcal{S}_{k+1,j+1} \\ \left\{k+1\right\}\in\mathcal{P}}} \prod_{M\in\mathcal{P}} \partial_M C\left(u_1^{1/n},\ldots,u_d^{1/n}\right) \\
&\quad \big.+ \sum_{j=1}^{k} \frac{n!}{(n-j)!}\cdot C\left(u_1^{1/n},\ldots,u_d^{1/n}\right)^{n-j}\cdot \sum_{\substack{\mathcal{P}\in\mathcal{S}_{k+1,j} \\ \left\{k+1\right\}\not\in\mathcal{P}}} \prod_{M\in\mathcal{P}} \partial_M C\left(u_1^{1/n},\ldots,u_d^{1/n}\right)
\Bigg\} =\left(*\right)_4.\\
\end{split}
\]
}
\hspace{-1.75mm}We perform an index shift in the first sum such that $j+1$ is replaced by $j$ and make use of the following two properties:
\begin{itemize}
\item[(a)] For all $\mathcal{P}\in\mathcal{S}_{l,1}=\left\{\left\{\left\{1,\ldots,l\right\}\right\}\right\}$ holds that $\left\{l\right\}\not\in \mathcal{P}$.
\item[(b)] For all $\mathcal{P}\in\mathcal{S}_{l,l}=\left\{\left\{\left\{1\right\},\ldots,\left\{l\right\}\right\}\right\}$ holds that $\left\{l\right\}\in \mathcal{P}$.
\end{itemize}
\hspace{-1.75mm}This results in
{\small
\[
\begin{split}
\left(*\right)_4&=\frac{1}{n^{k+1}}\left(\prod_{j=1}^{k+1} u_j\right)^{\frac{1}{n}-1} \Bigg\{ \sum_{j=1}^{k+1} \frac{n!}{(n-j)!} \cdot C\left(u_1^{1/n},\ldots,u_d^{1/n}\right)^{n-j} \Bigg.\\
&\quad\cdot \sum_{\substack{\mathcal{P}\in\mathcal{S}_{k+1,j} \\ \left\{k+1\right\}\in\mathcal{P}}} \prod_{M\in\mathcal{P}} \partial_M C\left(u_1^{1/n},\ldots,u_d^{1/n}\right) \\
&\quad \big.+ \sum_{j=1}^{k+1} \frac{n!}{(n-j)!}\cdot C\left(u_1^{1/n},\ldots,u_d^{1/n}\right)^{n-j}\cdot \sum_{\substack{\mathcal{P}\in\mathcal{S}_{k+1,j} \\ \left\{k+1\right\}\not\in\mathcal{P}}} \prod_{M\in\mathcal{P}} \partial_M C\left(u_1^{1/n},\ldots,u_d^{1/n}\right)
\Bigg\}\\
&=\frac{1}{n^{k+1}}\left(\prod_{j=1}^{k+1} u_j\right)^{\frac{1}{n}-1} \Bigg\{ \sum_{j=1}^{(k+1)\wedge n} \frac{n!}{(n-j)!} \cdot C\left(u_1^{1/n},\ldots,u_d^{1/n}\right)^{n-j} \Bigg.\\
&\quad\cdot \Bigg. \sum_{\mathcal{P}\in\mathcal{S}_{k+1,j}} \prod_{M\in\mathcal{P}} \partial_M C\left(u_1^{1/n},\ldots,u_d^{1/n}\right) \Bigg\},
\end{split}
\]
}
\hspace{-1.75mm}where we used the fact that $k+1=(k+1)\wedge n$ since $n>k$. This concludes the first case. Case 2 ($n\leq k$) is similar to the first one. The main difference is that $C\left(u_1^{1/n},\ldots,u_d^{1/n}\right)^{n-j}=1$ for $j=n$, which was not possible before since $j\leq k<n$. Now $k\wedge n=n$ and therefore we obtain
{
\small
\[
\begin{split}
\left(*\right)_2&=
\frac{1}{n^k}\left(\prod_{j=1}^k u_j\right)^{\frac{1}{n}-1} \Bigg\{ \sum_{j=1}^{n-1} \frac{n!}{(n-j)!}\cdot (n-j)\cdot C\left(u_1^{1/n},\ldots,u_d^{1/n}\right)^{n-j-1}\Bigg.\\
&\quad\cdot  \partial_{k+1} C\left(u_1^{1/n},\ldots,u_d^{1/n}\right)\frac{1}{n}u_{k+1}^{\frac{1}{n}-1}\cdot\sum_{\mathcal{P}\in\mathcal{S}_{k,j}} \prod_{M\in\mathcal{P}} \partial_M C\left(u_1^{1/n},\ldots,u_d^{1/n}\right) \\
&\quad \big.+ \sum_{j=1}^{n} \frac{n!}{(n-j)!} C\left(u_1^{1/n},\ldots,u_d^{1/n}\right)^{n-j} \frac{\partial}{\partial u_{k+1}}\left[\sum_{\mathcal{P}\in\mathcal{S}_{k,j}} \prod_{M\in\mathcal{P}} \partial_M C\left(u_1^{1/n},\ldots,u_d^{1/n}\right)\right]
\Bigg\}\\
&=  \frac{1}{n^{k+1}}\left(\prod_{j=1}^{k+1} u_j\right)^{\frac{1}{n}-1} \Bigg\{ \sum_{j=1}^{n-1} \frac{n!}{(n-(j+1))!} \cdot C\left(u_1^{1/n},\ldots,u_d^{1/n}\right)^{n-(j+1)} \Bigg.\\
&\quad\cdot \sum_{\substack{\mathcal{P}\in\mathcal{S}_{k+1,j+1} \\ \left\{k+1\right\}\in\mathcal{P}}} \prod_{M\in\mathcal{P}} \partial_M C\left(u_1^{1/n},\ldots,u_d^{1/n}\right) \\
&\quad \big.+ \sum_{j=1}^{n} \frac{n!}{(n-j)!}\cdot C\left(u_1^{1/n},\ldots,u_d^{1/n}\right)^{n-j}\cdot \sum_{\substack{\mathcal{P}\in\mathcal{S}_{k+1,j} \\ \left\{k+1\right\}\not\in\mathcal{P}}} \prod_{M\in\mathcal{P}} \partial_M C\left(u_1^{1/n},\ldots,u_d^{1/n}\right)
\Bigg\}\\
&=  \frac{1}{n^{k+1}}\left(\prod_{j=1}^{k+1} u_j\right)^{\frac{1}{n}-1} \Bigg\{ \sum_{j=1}^{n} \frac{n!}{(n-j)!} \cdot C\left(u_1^{1/n},\ldots,u_d^{1/n}\right)^{n-j} \Bigg.\\
&\quad \cdot \sum_{\substack{\mathcal{P}\in\mathcal{S}_{k+1,j} \\ \left\{k+1\right\}\in\mathcal{P}}} \prod_{M\in\mathcal{P}} \partial_M C\left(u_1^{1/n},\ldots,u_d^{1/n}\right) + \sum_{j=1}^{n} \frac{n!}{(n-j)!}\cdot C\left(u_1^{1/n},\ldots,u_d^{1/n}\right)^{n-j}\\
&\quad \big.\cdot \sum_{\substack{\mathcal{P}\in\mathcal{S}_{k+1,j} \\ \left\{k+1\right\}\not\in\mathcal{P}}} \prod_{M\in\mathcal{P}} \partial_M C\left(u_1^{1/n},\ldots,u_d^{1/n}\right)
\Bigg\}
\end{split}
\]
\[
\begin{split}
\quad&=\frac{1}{n^{k+1}}\left(\prod_{j=1}^{k+1} u_j\right)^{\frac{1}{n}-1} \Bigg\{ \sum_{j=1}^{(k+1)\wedge n} \frac{n!}{(n-j)!} \cdot C\left(u_1^{1/n},\ldots,u_d^{1/n}\right)^{n-j} \Bigg.\\
\quad&\quad\cdot \Bigg. \sum_{\mathcal{P}\in\mathcal{S}_{k+1,j}} \prod_{M\in\mathcal{P}} \partial_M C\left(u_1^{1/n},\ldots,u_d^{1/n}\right) \Bigg\},
\end{split}
\]
}where we applied Equation \ref{eq:claim2} in the second equality. In the third equality we performed an index shift in the first sum and used property (a) above. Since $n\leq k$ we have $n=(k+1)\wedge n$. This concludes the second case and hence the proof of Lemma \ref{lem:claim1}.
\end{proof}

Having proved the auxiliary lemma we can now easily prove the statement from Theorem \ref{thm:copuladensity}.
\begin{proof}[Proof of Theorem \ref{thm:copuladensity}]
Using Equation \ref{eq:maxcop} we obtain
\[
\begin{split}
c_{\Mb^{(n)}}(u_1,\ldots,u_d)&=\frac{\partial^d }{\partial u_1\cdots \partial u_d}C_{\Mb^{(n)}}(u_1,\ldots,u_d)=\frac{\partial^d }{\partial u_1\cdots \partial u_d}\left[C(u_1^{1/n},\ldots,u_d^{1/n})^n\right].
\end{split}
\]
Now Theorem \ref{thm:copuladensity} follows directly from Lemma \ref{lem:claim1} for $k=d$.
\end{proof}

\subsection*{Proof of Theorem \ref{thm:copuladeriv}}
\begin{proof}
Expression 4) follows directly by the definition of $c(u_1,u_2,u_3)$ (see Equation \ref{eq:pcc}). For expression 3(a) we can write 
\begin{equation*}
\begin{split}
	\partial_{23} C(u_1,&u_2,u_3) 
	\stackrel{Eq. (\ref{eq:pcc})}{=} c_{23}(u_2,u_3) \int_0^{u_1} c_{12}(v_1,u_2) c_{13;2}(C_{1|2}(v_1|u_2),C_{3|2}(u_3|u_2)) dv_1 \\
	&=  c_{23}(u_2,u_3)  \int_0^{u_1} \partial_{12} C_{12}(v_1,u_2) \partial_{13} C_{13;2}(C_{1|2}(v_1|u_2),C_{3|2}(u_3|u_2)) dv_1 \\
	&=  c_{23}(u_2,u_3) \int_0^{u_1} \frac{\partial}{\partial v_1}\bigg[ \underbrace{\partial_{2} C_{12}(v_1,u_2)}_{C_{1|2}(v_1|u_2)=w_1}\bigg] \partial_{13} C_{13;2}(C_{1|2}(v_1|u_2),C_{3|2}(u_3|u_2)) dv_1\\
	&=   c_{23}(u_2,u_3) \int_0^{u_1} \frac{\partial w_1}{\partial v_1} \frac{\partial}{\partial w_1}\partial_{3} C_{13;2}(w_1,C_{3|2}(u_3|u_2)) \bigg|_{\substack{w_1=C_{1|2}(v_1|u_2)}} dv_1 \\
	&=  c_{23}(u_2,u_3)  \partial_{3} C_{13;2}(C_{1|2}(u_1|u_2),C_{3|2}(u_3|u_2)).
	\end{split}
\end{equation*}
Expression 3(a) follows similarly. For expression 2(c) we have
\begin{equation*}
\begin{split}
\partial_{3} C(u_1,u_2,u_3) &= \int_0^{u_1}\int_0^{u_2} c(v_1,v_2,u_3) dv_1 dv_2 \\
	&=\int_0^{u_2} c_{23}(v_2,u_3) C_{1|23}(u_1|v_2,u_3) dv_2 \\
	&= \int_0^{u_2} c_{23}(v_2,u_3) \partial_{3} C_{13;2}(C_{1|2}(u_1|v_2),C_{3|2}(u_3|v_2)) dv_2.
	\end{split}
\end{equation*}
Further, $\partial_{3} C(u_1,u_2,u_3)$ can be written in another way which we use for the calculation of the copula $C(u_1,u_2,u_3)$. In particular
we have 
\begin{align}
\begin{split}\label{eq:d3C}
\partial_{3} C&(u_1,u_2,u_3) = \int_0^{u_1}\int_0^{u_2} c(v_1,v_2,u_3) dv_1 dv_2 \\
	&= \int_0^{u_2} \frac{\partial^2}{\partial v_2 \partial u_3} C_{23}(v_2,u_3) \frac{\partial}{\partial w_2} 
C_{13;2}(C_{1|2}(u_1|v_2),w_2) \bigg|_{w_2=C_{3|2}(u_3|v_2)} dv_2  \\
	&= \int_0^{u_2} \frac{\partial w_2}{\partial u_3}\frac{\partial}{\partial w_2} C_{13;2}(C_{1|2}(u_1|v_2),w_2) \bigg|_{w_2=C_{3|2}(u_3|v_2)} dv_2\\
	&= \int_0^{u_2} \frac{\partial}{\partial u_3} C_{13;2}(C_{1|2}(u_1|v_2),C_{3|2}(u_3|v_2)) dv_2  \\
	&= \frac{\partial}{\partial u_3} \left[ \int_0^{u_2} C_{13;2}(C_{1|2}(u_1|v_2),C_{3|2}(u_3|v_2)) dv_2 \right].
\end{split}
\end{align}
Similarly one can derive expression 2(a). For the copula in expression 1 we have
\begin{equation*}
\begin{split}
	C(u_1,u_2,u_3) &\stackrel{\eqref{eq:d3C}}{=} \int_0^{u_3} \frac{\partial}{\partial v_3} \left[ \int_0^{u_2} C_{13;2}(C_{1|2}(u_1|v_2),C_{3|2}(v_3|v_2)) dv_2 \right]  dv_3\\
	&= \int_0^{u_2} \left[ \int_0^{u_3} \frac{\partial}{\partial v_3}  C_{13;2}(C_{1|2}(u_1|v_2),C_{3|2}(v_3|v_2)) dv_3 \right] dv_2 \\
	&= \int_0^{u_2} C_{13;2}\left( C_{1|2}(u_1|v_2), C_{3|2}(u_3|v_2)\right) dv_2.
	\end{split}
\end{equation*}
Expression 2(b) follows by differentiating expression 1 with respect to $v_2$. It remains to show expression 3(b). For this we differentiate 2(c) with respect to $v_1$. Therefore we have
\begin{equation*}
\begin{split}
\partial_{13} C(u_1,&u_2,u_3) = \frac{\partial}{\partial u_1}
\left[ \int_0^{u_2} \partial_{3} C_{13;2}(C_{1|2}(u_1|v_2),C_{3|2}(u_3|v_2)) c_{23}(v_2,u_3) dv_2 \right] \\
& =  \int_0^{u_2} \partial_{13}
C_{13;2}(C_{1|2}(u_1|v_2),C_{3|2}(u_3|v_2))c_{12}(u_1,v_2)
 c_{23}(v_2,u_3) dv_2 \\
& = \int_0^{u_2} 
c_{13;2}
(C_{1|2}(u_1|v_2), C_{3|2}(u_3|v_2)) 
c_{12}(u_1,v_2)
 c_{23}(v_2,u_3) dv_2.
\end{split}
\end{equation*}
\end{proof}

\bibliographystyle{asa}
\bibliography{vine}

\end{document}